\pdfoutput=1
\documentclass[a4paper,UKenglish,cleveref, autoref, 
thm-restate]{lipics-v2021}
\title{Towards a Model Theory of Ordered Logics: Expressivity and Interpolation (Extended version)}
\titlerunning{Towards a Model Theory of Ordered Logics: Expressivity and Interpolation (Extended version)} 

\author{Bartosz Bednarczyk}
{Computational Logic Group, Technische Universit{\"a}t  Dresden, Germany \and 
Institute of Computer Science, University of Wroc\l aw, Poland
\and \url{https://bartoszjanbednarczyk.github.io/}}
{bartosz.bednarczyk@cs.uni.wroc.pl}
{https://orcid.org/0000-0002-8267-7554}
{supported by the ERC Consolidator Grant No. 771779 (DeciGUT).}

\author{Reijo Jaakkola}
{Tampere University, Finland
\and \url{https://reijojaakkola.github.io/}}
{reijo.jaakkola@tuni.fi}
{https://orcid.org/0000-0003-4714-4637}
{}

\authorrunning{B. Bednarczyk and R. Jaakkola}
\Copyright{Bartosz Bednarczyk and Reijo Jaakkola}


\relatedversion{Full version of this paper is available on arXiV.}
\nolinenumbers 
\acknowledgements{We would like to thank Antti Kuusisto and Jean Christoph Jung for fruitful discussions on related topics, as well as Tim Lyon and Emanuel Kieroński for language corrections.}

\EventNoEds{0}
\EventLongTitle{Accepted to the 47th International Symposium on Mathematical Foundations of Computer Science (MFCS)}
\EventShortTitle{MFCS 2022}
\EventAcronym{MFCS}
\EventYear{2022}
\EventDate{August 22—-26, 2022}
\EventLocation{Vienna, Austria}
\EventLogo{}
\SeriesVolume{0}
\ArticleNo{28}

\usepackage{microtype} 
\usepackage{ellipsis} 
\usepackage[abbreviations,british]{foreign} 

\usepackage{xcolor}
\usepackage{hyperref}

\definecolor{darkmidnightblue}{rgb}{0.0, 0.2, 0.4}
\definecolor{persianplum}{rgb}{0.44, 0.11, 0.11}

\hypersetup{
  colorlinks  = true, 
  citecolor   = persianplum, 
  urlcolor    = persianplum, 
  linkcolor   = persianplum
}

\usepackage{tikz}
\usepackage{todonotes}
\usepackage{xspace}

\usepackage{mathtools}
\usepackage{amssymb} 
\usepackage{amsmath}
\usepackage{amsthm}

\usepackage{caption}
\usepackage{float}
\usepackage{subcaption}

\usepackage[h]{esvect}

\makeatletter
\def\desclabel#1#2{\begingroup
\def\@currentlabel{#1}%
#1\label{#2}\endgroup
}
\makeatother

\usetikzlibrary{arrows, decorations.markings, shapes, calc}
\usepackage{nicolas-markey-expose}

\usetikzlibrary{shadows}
\tikzset{
diagonal fill/.style 2 args={fill=#2, path picture={
\fill[#1, sharp corners] (path picture bounding box.south west) -|
                         (path picture bounding box.north east) -- cycle;}},
reversed diagonal fill/.style 2 args={fill=#2, path picture={
\fill[#1, sharp corners] (path picture bounding box.north west) |- 
                         (path picture bounding box.south east) -- cycle;}}
}

\usetikzlibrary{hobby,backgrounds,calc,trees}

\usepackage{tikz-cd}
\usetikzlibrary{matrix}
\usepackage{stmaryrd}
\usepackage{stackengine}
\usepackage{xspace}
\usepackage{xcolor}
\usepackage[linesnumbered,ruled,vlined]{algorithm2e}

\SetCommentSty{mycommfont}
\SetKwInput{KwData}{Input}

\definecolor{ao(english)}{rgb}{0.0, 0.5, 0.0}
\definecolor{brickred}{rgb}{0.8, 0.25, 0.33}

\newcommand{\Logic}[1]{\ensuremath{\mathsf{#1}}} 
\newcommand{\logicL}{\Logic{L}} 
\newcommand{\Laffix}{\Logic{L}_{\mathsf{affix}}}   
\newcommand{\Linfix}{\Logic{L}_{\mathsf{inf}}}   
\newcommand{\Lsuffix}{\Logic{L}_{\mathsf{suf}}} 
\newcommand{\Lprefix}{\Logic{L}_{\mathsf{pre}}} 

\newcommand{\Gaffix}{\Logic{G}_{\mathsf{affix}}}   
\newcommand{\Ginfix}{\Logic{G}_{\mathsf{inf}}}     
\newcommand{\Gsuffix}{\Logic{G}_{\mathsf{suf}}}    
\newcommand{\Gprefix}{\Logic{G}_{\mathsf{pre}}}    

\newcommand{\GF}{\Logic{GF}}   
\newcommand{\GNFO}{\Logic{GNFO}}   
\newcommand{\UNFO}{\Logic{UNFO}}   
\newcommand{\FO}{\Logic{FO}}   

\newcommand{\complexityclass}[1]{\textsc{#1}} 
\newcommand{\NExpTime}{\complexityclass{NExpTime}} 
\newcommand{\Tower}{\complexityclass{Tower}}
\newcommand{\LogSpace}{\complexityclass{LogSpace}}
\newcommand{\PSpace}{\complexityclass{PSpace}}
\newcommand{\PTime}{\complexityclass{PTime}}

\newcommand{\str}[1]{{\mathfrak{#1}}}

\newcommand{\arity}{\mathsf{ar}}

\newcommand{\N}{{\mathbb{N}}}

\newcommand{\sqin}{%
  \mathrel{\vphantom{\sqsubset}\text{%
    \mathsurround=0pt
    \ooalign{$\sqsubset$\cr$-$\cr}%
  }}%
}

\newcommand{\relsymbol}[1]{\mathrm{#1}}
\newcommand{\relsymbolP}{\relsymbol{P}}
\newcommand{\relsymbolR}{\relsymbol{R}}
\newcommand{\relsymbolQ}{\relsymbol{Q}}
\newcommand{\relsymbolS}{\relsymbol{S}}
\newcommand{\relsymbolT}{\relsymbol{T}}

\newcommand{\relsymbolA}{\relsymbol{A}}
\newcommand{\relsymbolB}{\relsymbol{B}}

\newcommand{\relsymbolH}{\relsymbol{H}}
\newcommand{\sig}{\mathsf{sig}}

\newcommand{\domelem}[1]{\mathrm{#1}}                           
\newcommand{\domelema}{\domelem{a}}                             
\newcommand{\domelemb}{\domelem{b}}                             
\newcommand{\domelemc}{\domelem{c}}                             
\newcommand{\domelemd}{\domelem{d}}                             
\newcommand{\domeleme}{\domelem{e}}                             
\newcommand{\domelemf}{\domelem{f}}                               
\newcommand{\domelemg}{\domelem{g}}                             
\newcommand{\domelemh}{\domelem{h}}                             
\newcommand{\domelemw}{\domelem{w}}                             
\newcommand{\domelemtuplea}{\overline{\domelema}}                         
\newcommand{\domelemtupleb}{\overline{\domelemb}}                         
\newcommand{\domelemtuplec}{\overline{\domelemc}}                         
\newcommand{\domelemtupled}{\overline{\domelemd}}                         
\newcommand{\domelemtuplee}{\overline{\domeleme}}                         
\newcommand{\domelemtuplef}{\overline{\domelemf}}                           
\newcommand{\domelemtupleg}{\overline{\domelemg}}                         
\newcommand{\domelemtupleh}{\overline{\domelemh}}                         

\newcommand{\domelemtupledfromto}[2]{\overline{\domelemd}_{#1\ldots#2}}  
\newcommand{\domelemtupleefromto}[2]{\overline{\domeleme}_{#1\ldots#2}}  
\newcommand{\domelemtuplecfromto}[2]{\overline{\domelemc}_{#1\ldots#2}}  
\newcommand{\domelemtuplebfromto}[2]{\overline{\domelemb}_{#1\ldots#2}}  
\newcommand{\domelemtupleafromto}[2]{\overline{\domelema}_{#1\ldots#2}}  
\newcommand{\domelemtupleffromto}[2]{{\domelemtuplef}_{#1\dots#2}} 
\newcommand{\domelemtuplegfromto}[2]{{\domelemtupleg}_{#1\dots#2}} 

\newcommand{\var}[1]{\mathit{#1}}       
\newcommand{\varx}{\var{x}}             
\newcommand{\vary}{\var{y}}             
\newcommand{\varv}{\var{v}}             
\newcommand{\varu}{\var{u}}             
\newcommand{\vartuplex}{\overline{\varx}}    
\newcommand{\vartupley}{\overline{\vary}}                    
\newcommand{\vartuplexfromto}[2]{\overline{\varx}_{#1\ldots#2}}  
\newcommand{\vartupleyfromto}[2]{\overline{\vary}_{#1\ldots#2}}  

\newcommand{\tp}[3]{\mathsf{tp}^{#1}_{#2}(#3)}

\newcommand{\omegasat}[1]{\widehat{#1}}

\newcommand{\bisimulation}[1]{\mathcal{#1}} 
\newcommand{\bisimZ}{\bisimulation{Z}}
\newcommand{\bisimto}{\sim} 

\newcommand{\wit}{\mathsf{wit}}

\begin{document}

\maketitle

\begin{abstract}
  We consider the family of guarded and unguarded ordered logics, that constitute a recently rediscovered family of decidable fragments of first-order logic ($\FO$), in which the order of quantification of variables coincides with the order in which those variables appear as arguments of predicates.
  While the complexities of their satisfiability problems are now well-established, their model theory, however, is poorly understood. 
  Our paper aims to provide some insight into it.

  We start by providing suitable notions of bisimulation for ordered logics. 
  We next employ bisimulations to compare the relative expressive power of ordered logics, and to characterise our logics as bisimulation-invariant fragments of $\FO$ à la van Benthem.

  Afterwards, we study the Craig Interpolation Property~(CIP). 
  We refute yet another claim from the infamous work by Purdy, by showing that the fluted and forward fragments do not enjoy CIP. 
  We complement this result by showing that the ordered fragment and the guarded ordered logics enjoy CIP.
  These positive results rely on novel and quite intricate model constructions, which take full advantage of the ``forwardness'' of our logics.
\end{abstract}


\newcommand{\wrong}[1]{{\color{blue}\underline{#1}}}
\newcommand{\art}{\relsymbol{artist}}
\newcommand{\adm}{\relsymbol{admire}}
\newcommand{\bkpr}{\relsymbol{beekeeper}}
\newcommand{\env}{\relsymbol{envy}}
\newcommand{\student}{\relsymbol{student}}
\newcommand{\admires}{\relsymbol{admires}}
\newcommand{\prof}{\relsymbol{professor}}
\newcommand{\intro}{\relsymbol{introduce}}
\newcommand{\lecturer}{\relsymbol{lecturer}}
\newcommand{\person}{\relsymbol{person}}
\newcommand{\ispartof}{\relsymbol{isPartOf}}
\newcommand{\Narcissist}{\relsymbol{narcissist}}
\newcommand{\loves}{\relsymbol{loves}}
\newcommand{\haschild}{\relsymbol{hasChild}}
\newcommand{\hasparent}{\relsymbol{hasParent}}

\section{Introduction}\label{sec:introduction}
An ongoing research in computational logic has lead to discovery of new decidable fragments of first-order logics ($\FO$) that extend modal and description logics. 
The main ideas that were proposed in the past involve: restricting the number of variables~\cite{GradelKV97}, relativised quantification~\cite{AndrekaNB98,Benthem97}, restricted use of negation~\cite{SegoufinC13}, relativised negation~\cite{BaranyCS15}, one-dimensionality and uniformity~\cite{HellaK14}, separateness~\cite{0001VW16} and ordered quantification~\cite{Herzig90,Quine76}.
To compare aforementioned logics, the authors of~\cite[Section 4.7]{AndrekaNB98} proposed a list of desirable meta-properties of logic, which can serve as a yardstick to measure how ``nice'' a given logic is. 
We expect a logic $\logicL$ to
\begin{enumerate}[(A)]
  \item\label{intro:lab:A} be decidable and have the \emph{Finite Model Property} (FMP),
  \item\label{intro:lab:B} satisfy the Craig Interpolation Property (CIP), \ie for any $\logicL$-formulae $\varphi,\psi$ such that $\varphi \models \psi$ there should be an $\logicL$-formulae $\chi$, called an \emph{interpolant}, that uses only symbols appearing in the common vocabulary of $\varphi$ and $\psi$, so that $\varphi \models \chi \models \psi$ holds,
  \item\label{intro:lab:C} and to satisfy the analog of Łoś-Tarski Preservation Theorem (ŁTPT), \ie any $\logicL$-formula $\varphi$  preserved under substructures should be equivalent to some universal~$\logicL$-formula. 
 \end{enumerate}
It turned out that $\FO^2$ and $\GF$, example logics based on restricted number of variables and relativised quantification, are not ``nice'' as they do not enjoy CIP~\cite[Examples 1--2]{JungW21}.
In contrast, $\UNFO$ and $\GNFO$, the logics based on relativised negation, fulfil the properties (\ref{intro:lab:A})--(\ref{intro:lab:C}), consult:~\cite{SegoufinC13,BaranyBC18,BenediktCB16}.
For one-dimensionality, separateness and ordered quantification we have partial results only.

In this paper we take a closer look at logics enjoying ordered quantification, which have been receiving increasing attention recently~\cite{PrattHartmannS19,Bednarczyk21,Jaakkola21}. Their syntax can be informally explained as follows. We first require that all variables appearing in formulae are additionally indexed by the quantifier depth and then impose a certain restriction on such numbers in variable sequences in atoms.
Assuming that $\alpha(\vartuplex)$ is in the scope of the $n$-th quantifier (but not the $(n{+}1)$-th), in the fluted fragment $\Lsuffix$ of Quine~\cite{Quine76} (resp. in the ordered fragment $\Lprefix$ by Herzig~\cite{Herzig90}\footnote{Strictly speaking, the syntax of $\Lprefix$ is slightly more liberal than the original syntax of the ordered fragment as defined by Herzig, since the syntax of $\Lprefix$ allows requantifying variables.}) the tuple $\vartuplex$ is required to be a suffix (resp. a prefix) of the sequence $\varx_1, \varx_2, \ldots, \varx_n$.
The forward fragment $\Linfix$~\cite{Bednarczyk21} is more liberal and allows infixes in place of suffixes or prefixes.
An example formula $\varphi \in (\Lsuffix \cap \Linfix) \setminus \Lprefix$ is given below:
\begin{enumerate}\itemsep0em
\item No student admires every professor.
\vspace{-1em}
\[ 
\forall{\varx_1}\; (\student(\varx_1) \to \neg \forall{\varx_2} \;(\prof(\varx_2) \to \admires(\varx_1, \varx_2))) \]
\item No lecturer introduces any professor to every student.   
\vspace{-1em}
\[
\forall{\varx_1}\; \lecturer(\varx_1) \to \neg \exists{\varx_2}\;[\prof(\varx_2) \land \forall{\varx_3}\;(\student(\varx_3) \to \intro(\varx_1, \varx_2, \varx_3))]
\]
\end{enumerate}
Next, we provide a few coexamples, \ie formulae that, as stated, do not belong to any of $\Linfix, \Lprefix, \Lsuffix$. 
The \wrong{blue colour} indicates a mismatch in the variable ordering.
\begin{enumerate}\itemsep0em
    \item The relation $\ispartof$ is transitive.
    \vspace{-1em}
    \[
        \forall{\varx_1}\;\forall{\varx_2}\;\forall{\varx_3}\; \ispartof(\varx_1, \varx_2) \land \ispartof(\varx_2, \varx_3) \to \ispartof(\wrong{\varx_1, \varx_3})
    \]
    \item A narcissist is a person who loves himself.
    \vspace{-1em}
    \[
        \forall{\varx_1}\; \Narcissist(\varx_1) \to \person(\varx_1) \land \loves(\wrong{\varx_1,\varx_1}) 
    \]
    \item The binary relation $\haschild$ is the inverse of the $\hasparent$ relation.    
    \vspace{-1em}
    \[
        \forall{\varx_1}\;\forall{\varx_2}\; \hasparent(\varx_1, \varx_2) \leftrightarrow \hasparent(\wrong{\varx_2, \varx_1}) 
    \]
\end{enumerate}
All of $\Linfix, \Lsuffix, \Lprefix$ are decidable and have the Finite Model Property.
Their satisfiability problem is, respectively, $\Tower$-complete for $\Linfix$ and $\Lsuffix$, and $\PSpace$-complete for $\Lprefix$. 
Somehow unexpectedly, the $\Tower$-completeness of $\Lsuffix$ was established only recently by Pratt-Hartmann et al.~\cite{PrattHartmannS19}, after pointing out a mistake in the proof of the exponential-size model of $\Lsuffix$ by Purdy~\cite{Purdy02} and disproving Purdy's claim of $\NExpTime$-completeness of~$\Lsuffix$. 
The model theory of $\Linfix, \Lsuffix$, and $\Lprefix$ is, however, poorly understood. 
The only results that we are aware of are Purdy's claims that $\Lsuffix$ has CIP~\cite[Thm. 14]{Purdy02} and ŁTPT~\cite[Corr. 17]{Purdy02}.
But in the light of previously discovered errors, one should treat Purdy's paper with caution.

\subsection{Our results}\label{subsec:our-results}
This paper kick-starts a project of understanding the model theory of \emph{ordered logics}, by which we mean the logics $\Lprefix, \Lsuffix$, and $\Linfix$ as well as their intersections with the guarded fragment $\GF$~\cite{AndrekaNB98}, focusing on the problems mentioned in the introduction.

In~\cref{sec:expressive-power}, we design a suitable notion of bisimulations and compare the relative expressive power of ordered logics.
Our proofs employ standard model-theoretic constructions like the Compactness Theorem and $\omega$-saturated structures.
Next, we investigate CIP in~\cref{sec:interpolation}, which is the main technical contribution of the paper.
First, we focus on interpolation for the fluted and the forward fragments. 
We show that, surprisingly, $\Linfix$ and $\Lsuffix$ do not enjoy CIP, refuting yet another claim from the infamous work of Purdy~\cite[Thm. 14]{Purdy02}.
Fortunately, other members of the family of ordered logics enjoy CIP, as shown in Sections~\ref{subsec:CIP-for-Lprefix}--\ref{subsec:CIP-for-Gaffix}.
We stress here that standard techniques for proving CIP, \eg those based on zig-zag products~\cite{marx1995algebraic,HooglandMO99,BaranyBC18,Jaakkola22}, do not seem to work in our case.\footnote{For logics that are closed under negation on the level of formulas, zig-zag constructions seem to work only if the logics are \emph{one-dimensional} and \emph{uniform}, see~\cite{Jaakkola22} for more details. None of our logics are one-dimensional nor uniform.}
This forces us to take a different route: we construct models explicitly by specifying types of tuples.

We believe that our proof methods, which are based on novel and intricate model-theoretic constructions, are very general. 
In particular, we believe that our CIP proof for guarded ordered logics can serve as a useful meta-technique (or even a heuristic) for (dis)proving CIP for fragments of $\GF$. 
For instance, the proof can be adopted to fragments with CIP, deriving existing results (\eg for the $2$-variable $\GF$~\cite{HooglandMO99} or the uniform one-dimensional $\GF$~\cite{Jaakkola22}) and its failure gives hints why a certain fragment may not have CIP (\eg in the case of full~$\GF$).

\section{Preliminaries}\label{sec:preliminaries}
Henceforth, we employ standard terminology from (finite and classical) model theory~\cite{Libkin04,Hodges97}.
All the logics considered here will be fragments of the first-order logic ($\FO$) over purely-relational equality-free vocabularies, under the usual syntax and semantics. 

We fix a countably infinite set of variables $\{x_i \mid i \in \N\}$ and throughout this paper all the formulas use only variables from this set.
With $\sig(\varphi)$ we denote the set of relational symbols appearing in $\varphi$. 
We use $\arity(\relsymbolR)$ to denote the arity of $\relsymbolR$.
For a logic $\logicL$ and a signature $\sigma$ we use~$\logicL[\sigma]$ in place of~$\{ \varphi \in \logicL \mid \sig(\varphi) \subseteq \sigma \}$.
The $k$-variable fragment of $\logicL$ (\ie employing only the variables $\varx_1, \varx_2, \ldots, \varx_k$) is denoted~$\logicL^k$.
We write $\varphi(\vartuplex)$ to indicate that all free variables from $\varphi$ are members of $\vartuplex$. If $\vartuplex$ contains precisely the free
variables of $\varphi$, then we will emphasise this separately. 
Given a structure $\str{A}$ and $B \subseteq A$, we will use $\str{A} \upharpoonright B$ to denote the \emph{substructure} of $\str{A}$ that $B$ induces.\\

\noindent \textbf{Tuples and subsequences.}
An $n$-tuple is a tuple with $n$ elements. The $0$-tuple is denoted with $\epsilon$.
We use $\vartuplexfromto{i}{j}$ to denote the $(j{-}i{+}1)$-tuple $\varx_i, \varx_{i+1}, \ldots, \varx_j$.
We say that $\vartuplexfromto{i}{j}$ is an infix of a tuple $\vartuplexfromto{k}{l}$ if $k \leq i \leq j \leq l$ holds. 
If, in addition, $k = i$ (resp. $j = l$) we say that $\vartuplexfromto{i}{j}$ is a prefix (resp. suffix) of~$\vartuplexfromto{k}{l}$.
We use the word \emph{affix} as a place-holder for the words \emph{pre}fix, \emph{suf}fix or \emph{inf}ix.
For a set $S$, we write $\vartuplex \sqin S$ iff $\varx_i \in S$ for all indices $1 \leq i \leq |\vartuplex|$, where $|\vartuplex|$ denotes the length of~$\vartuplex$. 
A tuple $\domelemtuplea \sqin A$ is \emph{$\sigma$-live} in $\str{A}$ if $|\domelemtuplea| \leq 1$ or $\domelemtuplea \in \relsymbolR^{\str{A}}$ for some~$\relsymbolR \in \sigma$.\\

\noindent \textbf{Logics.}
We next introduce the logics $\Laffix \in \{ \Lprefix, \Lsuffix, \Linfix \}$. We start from $\Lsuffix$, which for technical reasons we need to define separately from $\Lprefix$ and $\Linfix$. For every $n \in \N$, we define the set $\Lsuffix(n)$ as follows:
\begin{itemize}\itemsep0em
    \item an atom $\alpha(\vartuplex)$ is in $\Lsuffix(n)$ if $\vartuplex$ is a suffix of~$\vartuplexfromto{1}{n}$,
    \item $\Lsuffix(n)$ is closed under Boolean connectives $\land, \lor, \neg, \to$,
    \item if $\varphi$ is in $\Lsuffix(n{+}1)$ then $\exists{\varx_{n{+}1}} \; \varphi$ and $\forall{\varx_{n{+}1}} \; \varphi$ are in $\Lsuffix(n)$.
\end{itemize}
We put $\Lsuffix := \Lsuffix(0)$, which is exclusively composed of sentences.

\hspace{-1.7em}To define the fragments $\logicL \in \{\Lprefix, \Linfix\}$, for every $n \in \N$ we define the set of $\logicL(n)$ as follows:
\begin{itemize}\itemsep0em
    \item an atom $\alpha(\vartuplex)$ is in $\Lprefix(n)$ if $\vartuplex$ is a prefix of~$\vartuplexfromto{1}{n}$ and in $\Linfix(n)$ if $\vartuplex$ is an infix of~$\vartuplexfromto{1}{n}$.
    \item if $\varphi \in \logicL(n_1)$ and $\psi \in \logicL(n_2)$, then for all $n\geq \max\{n_1,n_2\}$ we have that $\neg \varphi, (\varphi \lor \psi), (\varphi \land \psi), (\varphi \to \psi)$ are in $\logicL(n)$.
    \item if $\varphi$ is in $\logicL(n{+}1)$ then $\exists{\varx_{n{+}1}} \; \varphi$ and $\forall{\varx_{n{+}1}} \; \varphi$ are in $\logicL(n)$.
\end{itemize}
We set $\logicL := \logicL(0)$, which is exclusively composed of sentences. 
We stress that in contrast to $\Lsuffix$, the logics $\logicL \in \{\Lprefix, \Linfix\}$ allow us to requantify variables.
We recommend the reader to employ the above definition to show that $\forall{\varx_1}\forall{\varx_2}\forall{\varx_3} (\relsymbolR(\varx_1\varx_2\varx_3) \to (\relsymbolA(\varx_1) \land \exists{\varx_2}\exists{\varx_3} \relsymbolS(\varx_1\varx_2\varx_3))) \in \Linfix$.

Notice that if $\varphi(\vartuplex) \in \Laffix(n)$, where $\vartuplex$ lists all the free variables of $\varphi$ in order (with respect to their indices), then $\vartuplex$ is an affix of the tuple $\vartuplexfromto{1}{n}$.
The logics $\Lprefix, \Lsuffix$, and $\Linfix$ were studied under the names of ordered~\cite{Herzig90}, fluted~\cite{Quine76}, and forward~\cite{Bednarczyk21} fragments.
The \emph{guarded} counterparts $\Gaffix$ of~$\Laffix$, are defined as the intersection of $\Laffix$ and the guarded fragment~$\GF$~\cite{AndrekaNB98}, \ie by imposing that blocks of quantifiers are relativised by atoms (recalled below).
Abusing notation, we speak about all these logics collectively as \emph{ordered logics}.

\hspace{-1.75em} For reader's convenience we recall that $\GF$ is the smallest fragment of $\FO$  such that:
\begin{itemize}\itemsep0em
    \item Every atomic formula is in $\GF$;
    \item $\GF$ is closed under boolean connectives $\land, \lor, \neg, \to$;
    \item If $\varphi(\vartuplex, \vartupley)$ is in $\GF$ and $\alpha(\vartuplex, \vartupley)$ is an atom containing all free variables of $\varphi$ then both $\forall{\vartupley} \; (\alpha(\vartuplex, \vartupley) \to \varphi(\vartuplex, \vartupley))$ and $\exists{\vartupley} \; (\alpha(\vartuplex, \vartupley) \land \varphi(\vartuplex, \vartupley))$ are in $\GF$; 
    \item If $\varphi(\varx)$ has only a single free-variable $\varx$, then $\forall{\varx}\; \varphi$ and $\exists{\varx}\; \varphi$ are in $\GF$.
\end{itemize}
The atoms $\alpha$, appearing in the 3rd item of the above definition is called a \emph{guard}.

For a finite signature $\sigma$ and $n \in \N$, a \emph{$(\sigma, n)$-affix-type} is a conjunction of atoms with $n$ free variables $\vartuplexfromto{1}{n}$, in which for every $\relsymbolR \in \sigma$ and every affix $\vartuplexfromto{l}{k}$ of $\vartuplexfromto{1}{n}$, of length $\arity(\relsymbolR)$, exactly one of $\relsymbolR(\vartuplexfromto{l}{k})$, $\neg \relsymbolR(\vartuplexfromto{l}{k})$ appears as a conjunct. 
For a $\sigma$-structure $\str{A}$ and a tuple $\domelemtuplea \sqin A$ with $\tp{\Laffix[\sigma]}{\str{A}}{\domelemtuplea}$ we denote the \emph{unique} $(\sigma, |\domelemtuplea|)$-affix-type realised by $\domelemtuplea$ in $\str{A}$.


\newcommand{\menc}{\mathrm{menc}}

\subsection{Model Checking}\label{subsec:model-checking}

Before jumping into the main part of the paper, we would like to point out some results on the combined complexity of model checking problems of ordered logics, since these seem to be missing from the literature.
In what follows we will employ the \emph{matrix encoding of structure}, that is a standard encoding in finite model theory~\cite[p. 88]{Libkin04}. 
Given a $\{\relsymbolR_1,\ldots,\relsymbolR_m\}$-structure $\str{A}$ with a linearly-ordered domain $A$, by its \emph{matrix encoding} we mean a binary string $\menc(\str{A}) := 0^n 1\menc(\relsymbolR_1)\ldots\menc(\relsymbolR_m)$, where $\menc(\relsymbolR_i)$ is a binary sequence of length $|A|^{\arity(\relsymbolR_i)}$, in which the $j$-th bit is $1$ iff the $j$-th tuple in the lexicographic ordering of $|A|^{\arity(\relsymbolR_i)}$ belongs to $\relsymbolR_i^\str{A}$. 

The following theorem collects our complexity results. 
We have not tried to optimise the upper bounds for $\Gprefix$ and $\Lprefix$: it is quite possible that they can be improved further.

\begin{theorem}
	Under the matrix encoding of structures, the combined complexity of the model-checking problem for a logic $\logicL$ is 
	\begin{enumerate}\itemsep0em
		\item decidable in $\PTime$ for $\Gprefix$ and $\Lprefix$,
		\item $\PTime$-complete for $\logicL \in \{\Gsuffix, \Ginfix, \Lsuffix\}$, and
		\item $\PSpace$-complete for $\logicL = \Linfix$.
	\end{enumerate}
\end{theorem}
\begin{proof}
	The upper bound for $\Gprefix$ follows from the second item while the upper bound for $\Lprefix$ is proved in~\cref{appendix:model-checking-fluted-logic}. For the second item, the lower bound follows for all of the logics from the fact that
	they embed standard modal logic, for which the combined complexity is $\PTime$-complete~\cite[Cor.~3.1.7]{GradelKLMSVVW07}. For $\Gsuffix$ and $\Ginfix$ matching upper bounds follow from the fact that
	the combined complexity of the guarded fragment is $\PTime$-complete, while for $\Lsuffix$ the matching upper bound is proved in~\cref{appendix:model-checking-fluted-logic}.
	Finally, for the third item, the upper bound follows from the fact that the combined complexity of $\FO$ is $\PSpace$-complete~\cite{BerwangerG01}, while the matching lower bound follows
	from the fact that $\Linfix$ contains monadic $\FO$, for which the combined complexity of model-checking is known to be $\PSpace$-complete~\cite[p.~99]{Libkin04}.
\end{proof}

The matrix encoding is not the only natural way of encoding models. Another option would be to use the \emph{list}/\emph{database encoding} of models, where one essentially encodes relations by
listing the tuples that they contain, as opposed to describing their adjacency matrices. It is easy to see that, if there is no bound on the arities of the relation symbols, then
the list encoding of a model can be exponentially more succinct than its matrix encoding. Our proofs for the upper bounds of $\Lprefix$ and $\Lsuffix$ are heavily dependent on
the fact that we are using the matrix encoding of models, and hence it is conceivable that the complexities are higher if we are using list encoding.\footnote{They can not decrease, because a list encoding of a model can always be constructed efficiently from its matrix encoding.}
We leave the related investigations as a very interesting future research direction.
\section{Expressive power}\label{sec:expressive-power}
We study the relative expressive power of ordered logics with a suitable notion of bisimulations.
\vspace{-\baselineskip}
\begin{definition}\label{def:Laffix-bisimulations}
  A non-empty set $\bisimZ \subseteq \bigcup_{n < \omega} (A^n \times B^n)$
  is a \emph{$\Laffix[\sigma]$-bisimulation} between pointed structures $\str{A},\domelemtuplea$ and $\str{B},\domelemtupleb$, where $|\domelemtuplea|=|\domelemtupleb|$, if and only if $(\domelemtuplea, \domelemtupleb) \in \bisimZ$ and for all~$(\domelemtuplec, \domelemtupled) \in \bisimZ$ the following conditions hold:
  \begin{description} \itemsep0em
    \item[\desclabel{(atomic harmony)}{bisim:atomic-harmony}] $\tp{\Laffix[\sigma]}{\str{A}}{\domelemtuplec} = \tp{\Laffix[\sigma]}{\str{B}}{\domelemtupled}$.
    \item[\desclabel{(forth)}{bisim:forth}]  For a (possibly empty) affix $\domelemtuplecfromto{i}{j}$ of $\domelemtuplec$ and $\domeleme\in A$ there is $\domelemf \in B$ s.t.~$(\domelemtuplecfromto{i}{j}\domeleme, \domelemtupledfromto{i}{j}\domelemf) \in \bisimZ$.
    \item[\desclabel{(back)}{bisim:back}]  For a (possibly empty) affix $\domelemtupledfromto{i}{j}$ of $\domelemtupled$ and $\domelemf\in B$ there is $\domeleme \in A$ s.t.~$(\domelemtuplecfromto{i}{j}\domelemd, \domelemtupledfromto{i}{j}\domelemf) \in \bisimZ$.
  \end{description}
\end{definition}

For $\Gaffix$, we replace the conditions \ref{bisim:forth}, \ref{bisim:back} by their guarded counterparts:
\begin{description} \itemsep0cm 
  \item[\desclabel{(gforth)}{bisim:gforth}] For a (possibly empty) affix $\domelemtuplecfromto{i}{j}$ of $\domelemtuplec$ and a $\sigma$-live tuple $\domelemtuplee$ in $\str{A}$ such that $\domelemtuplecfromto{i}{j} = \domelemtupleefromto{1}{j{-}i{+}1}$ there is a $\sigma$-live tuple $\domelemtuplef$ with $\domelemtupledfromto{i}{j} = \domelemtupleffromto{1}{j{-}i{+}1}$ and $(\domelemtuplee, \domelemtuplef) \in \bisimZ$,
  \item[\desclabel{(gback)}{bisim:gback}] For a (possibly empty) affix $\domelemtupledfromto{i}{j}$ of $\domelemtupled$ and a $\sigma$-live tuple $\domelemtuplef$ in $\str{B}$ such that $\domelemtupledfromto{i}{j} = \domelemtupleffromto{1}{j{-}i{+}1}$ there is a $\sigma$-live tuple $\domelemtuplee$ with $\domelemtuplecfromto{i}{j} = \domelemtupleefromto{1}{j{-}i{+}1}$ and $(\domelemtuplee, \domelemtuplef) \in \bisimZ$,
\end{description} 

For a logic $\logicL$ and a finite signature $\sigma$, we write $\str{A} \equiv_{\logicL[\sigma]} \str{B}$ if $\str{A}$ and $\str{B}$
 satisfy the same $\logicL[\sigma]$-sentences, and we write $\str{A} \bisimto_{\logicL[\sigma]} \str{B}$ if there is an $\logicL[\sigma]$-bisimulation between $\str{A}$ and $\str{B}$. 
 If $|\domelemtuplea|=|\domelemtupleb|$, we use $\str{A}, \domelemtuplea \equiv_{\logicL[\sigma]} \str{B}, \domelemtupleb$ to denote that for every (possibly empty) affix $\domelemtupleafromto{i}{j}$ of $\domelemtuplea$ and $\varphi(\vartuplexfromto{i}{j}) \in \logicL[\sigma]$, where $\vartuplexfromto{i}{j}$ is an affix of $(\varx_1,\dots,\varx_n)$, we have that
$\str{A} \models \varphi(\domelemtupleafromto{i}{j})$ if and only if $\str{B} \models \varphi(\domelemtuplebfromto{i}{j})$. For the next lemma consult~\cref{appendix:lemma:linking-bisimilarity-and-equivalence}. 

\begin{lemma}\label{lemma:linking-bisimilarity-and-equivalence}
  Let $\logicL \in \{ \Laffix, \Gaffix \}$. Then $\str{A}, \domelemtuplea \bisimto_{\logicL[\sigma]} \str{B}, \domelemtupleb$ implies $\str{A}, \domelemtuplea \equiv_{\logicL[\sigma]} \str{B}, \domelemtupleb$.
  The converse holds over $\omega$-saturated $\str{A}$ and $\str{B}$.
\end{lemma}



A logic $\logicL_2$ is \emph{at least as expressive as} a logic $\logicL_1$ (written $\logicL_1 \preceq \logicL_2$) if for all $\varphi \in \logicL_1$ there is a $\psi \in \logicL_2$ such that~$\varphi \equiv \psi$.
We write $\logicL_1 \approx \logicL_2$ iff $\logicL_1 \preceq \logicL_2$ and $\logicL_2 \preceq \logicL_1$.
In case $\logicL_1 \not\preceq \logicL_2$ and $\logicL_2 \not\preceq \logicL_1$ we call $\logicL_1$ and $\logicL_2$  incomparable.
Lastly, $\logicL_1 \prec \logicL_2$ denotes that $\logicL_2$ is \emph{strictly more expressive} than $\logicL_1$, \ie $\logicL_1 \preceq \logicL_2$ and $\logicL_1 \not \approx \logicL_2$.
Note that, by definition, all the considered fragments $\logicL$ satisfy $\logicL \preceq \Linfix$ and $\logicL \prec \FO$ (every such $\logicL$ is decidable).
Moreover, $\Gaffix \prec \Laffix$ is a consequence of $\forall{\varx_1}\forall{\varx_2} \relsymbolR(\varx_1, \varx_2)$ not being $\GF[\{ \relsymbolR \}]$-definable (which is well-known and follows from the fact that $\GF$ has the tree-model property).
Our results are as follows:

\begin{theorem}\label{thm:expressive-power-full-characterisation}
  (a) $\Lprefix \prec \Lsuffix {\approx} \Linfix \prec \FO$,
  (b) $\Gaffix \prec \Laffix$ for all affixes,
  (c) $\Gsuffix \prec \Ginfix$,
  (d) $\Gprefix \prec \Ginfix$, and
  (e) otherwise the logics are incomparable.
\end{theorem}
\begin{proof}
  Full proofs are in~\cref{appendix:thm:expressive-power-full-characterisation}.
  The relationships between different logics with separating examples (omitting trivial examples due to guardedness) are depicted below. 
  With $\varphi_{\textit{pre}}$ we denote the formula $\forall{\varx_1\varx_2\varx_3}\ \relsymbolR(\varx_1\varx_2\varx_3) \to \relsymbolS(\varx_1\varx_2)$, while $\varphi_{\textit{suf}}$ denotes $\forall{\varx_1\varx_2\varx_3}\ \relsymbolR(\varx_1\varx_2\varx_3) \to \relsymbolT(\varx_2\varx_3)$.
  Solid (resp. dashed) arrows from $\logicL_1$ to $\logicL_2$ denote that $\logicL_1 \prec \logicL_2$ holds (resp. that the logics are incomparable).
  
  
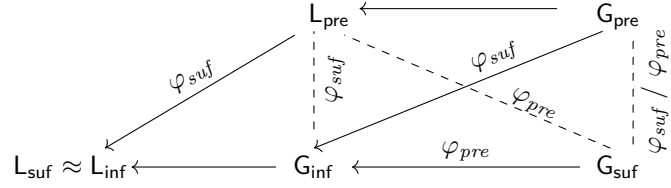
\vspace{-1em} 
  \begin{figure}[H]
    \centering
    \begin{tikzpicture}[transform shape]
        \node[] at (-0.2, 0) {$\Lsuffix \approx \Linfix$};
        \node[] at (3.2, 2) {$\Lprefix$};
        \node[] at (7, 2) {$\Gprefix$};
        \node[] at (3, 0) {$\Ginfix$};
        \node[] at (7, 0) {$\Gsuffix$};

        \path[->] (2.75,1.75) edge node[rotate=30, yshift=6] {$\varphi_{\textit{suf}}$} (0.25,0.25); 
        \path[->] (2.5, 0) edge node[yshift=6] {} (0.6, 0); 
        \path[->] (6.2, 2.1) edge node[yshift=6] {} (3.6, 2.1); 
        \path[->] (6.5, 0) edge node[yshift=6] {$\varphi_{\textit{pre}}$} (3.5, 0); 

        \path[dashed] (7.2,0.3)  edge node[rotate=90, yshift=-8] {$\varphi_{\textit{suf}}$ / $\varphi_{\textit{pre}}$} (7.2, 1.7); 

        \path[dashed] (6.9,0.25)  edge node[rotate=-23.5, yshift=4] {\; \; \; \; \; \; \; $ \varphi_{\textit{pre}}$}  (3.4, 1.75); 

        \path[dashed] (3, 0.5)  edge node[rotate=90, yshift=-8] {$ \varphi_{\textit{suf}}$}  (3, 1.8); 

        \path[->] (6.8, 1.8)  edge node[rotate=25, yshift=5.5, xshift=-1] {\; \; \; \; \; \; $\varphi_{\textit{suf}}$} (3.0, 0.25); 
    \end{tikzpicture}
  \end{figure}
  \vspace{-1em} 
  \noindent The equi-expressivity of $\Linfix$ and $\Lsuffix$ is an easy observation: we turn each maximally nested subformulae into DNF and push the atoms violating the definition of $\Laffix$ outside.
\end{proof}

\noindent Knowing the relative expressive power of our logics, we would like to characterise them as bisimulation-invariant fragments of $\FO$, as was done with other decidable logics, see \eg~\cite{Gradel014}.
\noindent Given a formula $\varphi(\vartuplex) \in \logicL$, we say that it is $\bisimto_{\logicL}$-invariant iff for all 
$\str{A},\domelemtuplea \bisimto_{\logicL}^{\sig(\varphi)} \str{B}, \domelemtupleb$
 we have $\str{A} \models \varphi(\domelemtuplea) \Leftrightarrow \str{B} \models \varphi(\domelemtupleb)$. $\logicL$ is $\bisimto_{\logicL}$-invariant iff all its formulae are $\bisimto_{\logicL}$-invariant. We will next show that $\Laffix$ (resp. $\Gaffix$) are exactly the $\bisimto_{\Laffix}$- (resp. $\bisimto_{\Gaffix}$-) invariant fragments of~$\FO$. 
This confirms that our notion of bisimulation is the right~one.

\begin{theorem}\label{thm:van-benthem}
 Let $\logicL \in \{ \Laffix, \Gaffix \}$ and let $\varphi(\vartuplex)$ be a $\bisimto_\logicL$-invariant $\FO$ formula. 
 Then there exists a formula $\psi(\vartuplex)$ in $\logicL$ which is equivalent with $\varphi(\vartuplex)$.
\end{theorem}
\begin{proof}
  We follow standard proof methods, see \eg \cite[Thm. 3.2]{BaranyBC18}.
  Suppose $\varphi(\varx_1,\dots,\varx_n) \in \FO$ is $\bisimto_\logicL$-invariant, where $\vartuplex = (\varx_1,\dots,\varx_n)$ enumerates precisely the set of free variables of~$\varphi$. The case when $\varphi$ is unsatisfiable $\varphi$ is trivial, thus assume otherwise.
  Consider the set 
  $\Gamma := \{\chi(\vartuplexfromto{i}{j}) \in \logicL \mid \varphi(\vartuplex) \models \chi(\vartuplexfromto{i}{j})\}$.
  Clearly $\varphi(\vartuplex) \models \Gamma$.
  Since $\FO$ is compact, it suffices to show that $\Gamma \models \varphi(\vartuplex)$. Let $\str{A}$ be a structure and $ \domelemtuplea \in A^n$ so that $\str{A} \models \chi(\domelemtupleafromto{i}{j})$, for every $\chi(\vartuplexfromto{i}{j}) \in \Gamma$.
  Next, consider the set $\Sigma := \{\chi(\vartuplexfromto{i}{j}) \in \logicL \mid \str{A} \models \chi(\domelemtupleafromto{i}{j})\}$.
  Again, by compactness of $\FO$ we can show that $\Sigma \cup \{\varphi\}$ is consistent. 
  Take a structure $\str{B}$ and $\domelemtupleb \in B^n$ so that $\str{B} \models \varphi(\domelemtupleb)$ and $\str{B} \models \chi(\domelemtuplebfromto{i}{j})$, for every $\chi(\vartuplexfromto{i}{j}) \in \Sigma$. 
  Observe that by construction $\str{A}, \domelemtuplea \equiv_{\logicL} \str{B}, \domelemtupleb$. 
  Replacing $\str{A}$ and $\str{B}$ with their $\omega$-saturated elementary extensions $\hat{\str{A}}$ and $\hat{\str{B}}$, we know by~\cref{lemma:linking-bisimilarity-and-equivalence} that $\hat{\str{A}}, \domelemtuplea \bisimto_{\logicL} \hat{\str{B}}, \domelemtupleb$.
  Chasing the~resulting~diagram~we~get~$\str{A} \models \varphi(\domelemtuplea)$.
\end{proof}


\section{Craig Interpolation}\label{sec:interpolation}

Recall that the Craig Interpolation Property (CIP) for a logic $\logicL$ states that if $\varphi(\vartuplex) \models \psi(\vartuplex)$ holds (with $\varphi$ and $\psi$ having the same free variables), then there is a $\chi(\vartuplex) \in \logicL[\sig(\varphi) \cap \sig(\psi)]$ (an $\logicL$-\emph{interpolant}) such that $\varphi(\vartuplex) \models \chi(\vartuplex)$ and $\chi(\vartuplex) \models \psi(\vartuplex)$~hold. We always assume that both $\varphi$ and $\psi$ are satisfiable, otherwise we can take $\bot$ as a trivial interpolant.


To reason about interpolants we employ the notion of \emph{joint consistency}~\cite{Robinson1960}. 
We say that $\logicL$-formulae $\varphi(\vartuplexfromto{1}{n})$ and $\psi(\vartuplexfromto{1}{n})$ (having exactly $\vartuplexfromto{1}{n}$ free) are \emph{jointly-$\logicL[\tau]$-consistent} (or just \emph{jointly consistent} in case $\tau := \sig(\varphi) \cap \sig(\psi)$ and $\logicL$ are known from the context), if there are structures $\str{A} \models \varphi(\domelemtuplea)$ and $\str{B} \models \psi(\domelemtupleb)$ such that $\str{A}, \domelemtuplea \bisimto_{\logicL[\tau]} \str{B}, \domelemtupleb$.
The next lemma is classic and links joint consistency and interpolation: see~\cref{appendix:lemma:joint-consistency-vs-interpolation}.
\begin{lemma}\label{lemma:joint-consistency-vs-interpolation}
  Let $\logicL \subseteq \FO$, and let $\varphi(\vartuplex), \psi(\vartuplex) \in \logicL$ with $\tau := \sig(\varphi) \cap \sig(\psi)$. 
  Then $\varphi(\vartuplex)$ and $\neg \psi(\vartuplex)$ are jointly consistent iff there is no $\logicL[\tau]$-interpolant for $\varphi(\vartuplex) \models \psi(\vartuplex)$.
\end{lemma}

We simplify the reasoning about ordered logics by employing suitable normal forms.
We say that a formula $\varphi(\vartuplex)$ from\footnote{To avoid notational glitter we will be a bit careless when dealing with formulae with free-variables.} $\Lprefix$ (resp. from $\Gaffix$) is in \emph{normal form} if it has the shape:
    \begin{description} \itemsep0em
    \item[\desclabel{(NForm-$\Lprefix$)}{NForm-Lprefix}] 
    $
    \relsymbolH(\vartuplex) \; \land \; 
    \bigwedge_{i=1}^{s} \forall{\vartuplexfromto{1}{\ell_i}} (\alpha_i \to \exists \varx_{\ell_i{+}1} \beta_i) \; \land \;
    \bigwedge_{j=1}^{t} \forall{\vartuplexfromto{1}{\ell_j}} (\alpha_j \to \forall \varx_{\ell_j{+}1} \beta_j)
    $,
    \item[\desclabel{(NForm-$\Gaffix$)}{NForm-Gaffix}] 
    $
    \relsymbolH(\vartuplex) \; \land \;
    \bigwedge_{i=1}^{s} \forall{\vartuplexfromto{1}{\ell_i}} (\relsymbolR_i( \vartuplexfromto{1}{\ell_i}) \to 
    \exists \vartuplexfromto{\ell_i{+}1}{\ell_i {+} k_i} (\relsymbolS_i(\vartuplexfromto{1}{\ell_i {+} k_i}) \land \psi_i(\vartuplexfromto{1}{\ell_i {+} k_i}))) \; \land \;
    \bigwedge_{j=1}^{t} \forall{\vartuplexfromto{1}{\ell_j}} (\relsymbolR_j( \vartuplexfromto{1}{\ell_j}) {\to} \psi_j(\vartuplexfromto{1}{\ell_j}) {\to} \forall{\vartuplexfromto{\ell_j{+}1}{\ell_j'}} (\relsymbolT_j(\vartuplexfromto{1}{\ell_j'}) \to \psi_j'(\vartuplexfromto{1}{\ell_j'}))),
    $
    \end{description}
    where $\alpha_i, \alpha_j, \beta_i$ and $\beta_j$ are quantifier-free $\Lprefix$-formulae, $\relsymbolR_i, \relsymbolR_j$, $\relsymbolT_j$ and $\relsymbolH$ are relational symbols, and $\psi_i, \psi_j$ and $\psi_j'$ are $\Gaffix$-formulae.
    The symbol $\relsymbolH$ is called the \emph{head} of $\varphi(\vartuplex)$.
    We will often speak about \emph{existential/universal requirements} of a formula, meaning the appropriate subformulae with the maximal quantifier prefix $\forall^*\exists^*$ and $\forall^*$. 
    In aforementioned normal forms we implicitly allow various parameters to be zero, \eg in subformulae of the form $\forall{\vartuplexfromto{1}{\ell_i}} (\alpha_i \to \exists \varx_{\ell_i{+}1} \beta_i)$ we allow $\ell_i = 0$, and we agree that the result is $\exists \varx_{\ell_i{+}1} \beta_i$.

    The following lemma can be shown using standard renaming techniques, in complete analogy to~\cite{Bednarczyk21,Jaakkola21}, with a minor (but technically tedious) modification in the case of~$\Gsuffix$, see~\cref{appendix:sec:normal-forms}.
    \begin{lemma}\label{lemma:normal-forms}
    Let $\logicL \in \{\Lprefix, \Gaffix\}$, and take $\varphi(\vartuplex), \psi(\vartuplex) \in \logicL$. Suppose that there are models $\str{A}$ and $\str{B}$ such that $\str{A} \models \varphi(\domelemtuplea)$, $\str{B} \models \psi(\domelemtupleb)$ and $\str{A},\domelemtuplea \bisimto_{\logicL[\tau]}
 \str{B},\domelemtupleb$, where $\tau = \sig(\varphi) \cap \sig(\psi)$. Then there exist formulae $\varphi'(\vartuplex),\psi'(\vartuplex) \in \logicL$ in normal form and extensions $\str{A}'$ and $\str{B}'$ of $\str{A}$ and $\str{B}$ respectively, such that (i) $\varphi'(\vartuplex)$ and $\psi'(\vartuplex)$ have the same head $\relsymbolH$, (ii) $\sig(\varphi') \cap \sig(\psi') = \tau \cup \{ \relsymbolH\}$, (iii) $\varphi'(\vartuplex) \models \varphi(\vartuplex)$ and $\psi'(\vartuplex) \models \psi(\vartuplex)$, and (iii) $(\str{A}',\domelemtuplea) \bisimto_{\logicL[\tau \cup \{\relsymbolH\}]} (\str{B}',\domelemtupleb)$ holds.
    \end{lemma}

    The following lemma is a useful tool when dealing with interpolation, allowing us to switch our attention to a certain satisfiability problem. 
    Its proof is routine, consult~\cref{appendix:lemma:aux-lemma-for-interpolation}.

\begin{lemma}\label{lemma:aux-lemma-for-interpolation}
  Let $\logicL \in \{ \Lprefix, \Gaffix \}$. 
  If for any jointly-consistent $\logicL$-formulae $\varphi(\vartuplex), \psi(\vartuplex)$ in normal forms from~\cref{lemma:normal-forms} with the same head, there is $\str{U} \models \varphi(\vartuplex) \land \psi(\vartuplex)$, then $\logicL$ has CIP.\@
\end{lemma}

\subsection{Disproving CIP in $\Linfix$ and $\Lsuffix$}

We start our investigation of CIP for $\Laffix$ and $\Gaffix$ by further discrediting the infamous work of Purdy~\cite{Purdy02}. 
We prove, in stark contrast to~\cite[Thm. 14]{Purdy02}, that $\Lsuffix$ does not have~CIP. 

\begin{theorem}\label{thm:FL-and-FF-doesnt-have-CIP}
  $\Linfix$ and $\Lsuffix$ do not have CIP. More specifically, there are $\Lsuffix^2$-sentences $\varphi, \psi$ with $\varphi \models \psi$ but without any $\Linfix[\sig(\varphi) \cap \sig(\psi)]$-interpolant.
\end{theorem}
\begin{proof}
  Consider the following $\Linfix^3$-sentences $\varphi$ and $\psi$, presented respectively below:
  \begin{align*}
  \hspace{-1.5em}\forall{\vartuplexfromto{1}{3}}[(\relsymbolR(\varx_1,\varx_2) \land \relsymbolR(\varx_2,\varx_3)) \to (\relsymbolP_1(\varx_1) \land \relsymbolP_2(\varx_3))] \; \land \; \forall{\varx_1}\forall{\varx_2}[(\relsymbolP_1(\varx_1) \land \relsymbolP_2(\varx_2)) \to \relsymbolR(\varx_1,\varx_2)]\\
  \hspace{-1em}\exists{\vartuplexfromto{1}{3}} [\relsymbolR(\varx_1,\varx_2) \land \relsymbolR(\varx_2,\varx_3) \land \relsymbolQ_1(\varx_1) \land \relsymbolQ_2(\varx_3)] \; \land \;  \forall{\varx_1}\forall{\varx_2}[(\relsymbolQ_1(\varx_1) \land \relsymbolQ_2(\varx_2)) \to \neg\relsymbolR(\varx_1,\varx_2)],
  \end{align*}
  with $\str{A} \models \varphi$ iff $(\relsymbolR^{\str{A}} \circ \relsymbolR^{\str{A}}) \subseteq \relsymbolP_1^{\str{A}} \times \relsymbolP_2^{\str{A}} $ and $\relsymbolP_1^{\str{A}} \times \relsymbolP_2^{\str{A}} \subseteq \relsymbolR^{\str{A}}$,
  and $\str{B} \models \psi$ iff $\relsymbolQ_1^{\str{B}} \times \relsymbolQ_2^{\str{B}} \subseteq B^2 \setminus \relsymbolR^{\str{B}} $ and there are $(\domelema, \domelemb), (\domelemb, \domelemc) \in \relsymbolR^{\str{B}}$ with $\domelema \in \relsymbolQ_1^{\str{B}}$ and $\domelemc \in \relsymbolQ_2^{\str{B}}$.
  Observe that $\varphi \models \neg \psi$, since $\varphi$ entails transitivity of $\relsymbolR$, while $\psi$ entails that this is not the case. But $\varphi$ and $\psi$ are jointly-$\Linfix[\{ \relsymbolR \}]$-consistent (it suffices to take $\str{A}$ and $\str{B}$ depicted below, \cf~\cref{appendix:thm:FL-and-FF-doesnt-have-CIP}). 
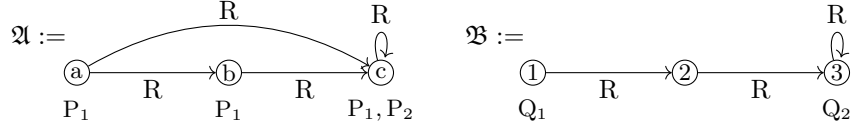
  \begin{figure}[H]
    \centering
    \begin{tikzpicture}[transform shape]
        \node[] at (1.5, 0.5) {$\str{A} :=$};
        \draw (2, 0) node[ptrond, label=center:\small{$\domelema$}] (a) {};
        \draw (4, 0) node[ptrond, label=center:\small{$\domelemb$}] (b) {};
        \draw (6, 0) node[ptrond, label=center:\small{$\domelemc$}] (c) {};

        \path[->] (a) edge node[yshift=-6] {\( \relsymbolR \)} (b); 
        \path[->] (b) edge node[yshift=-6] {\( \relsymbolR \)} (c); 
        \path[->] (a) edge[bend left=30] node[yshift=6] {\( \relsymbolR \)} (c); 
        \path[->] (c) edge[loop above] node[] {\( \relsymbolR \)} (c); 

        \node[] at (7.5, 0.5) {$\str{B} :=$};
        \draw (8, 0) node[ptrond, label=center:\small{$1$}] (one) {};
        \draw (10, 0) node[ptrond, label=center:\small{$2$}] (two) {};
        \draw (12, 0) node[ptrond, label=center:\small{$3$}] (three) {};

        \path[->] (one) edge node[yshift=-6] {\( \relsymbolR \)} (two); 
        \path[->] (two) edge node[yshift=-6] {\( \relsymbolR \)} (three); 
        \path[->] (three) edge[loop above] node[] {\( \relsymbolR \)} (three); 

        \draw (2, -0.5) node[label=center:\small{$\relsymbolP_1$}] (below_a) {};
        \draw (4, -0.5) node[label=center:\small{$\relsymbolP_1$}] (below_b) {};
        \draw (6, -0.5) node[label=center:\small{$\relsymbolP_1, \relsymbolP_2$}] (below_c) {};
        \draw (8, -0.5) node[label=center:\small{$\relsymbolQ_1$}] (below_one) {};
        \draw (12, -0.5) node[label=center:\small{$\relsymbolQ_2$}] (below_three) {};
    \end{tikzpicture}
  \end{figure}
  Hence, by~\cref{lemma:joint-consistency-vs-interpolation} there is no $\Linfix[\{ \relsymbolR \}]$-interpolant for $\varphi \models \neg \psi$.
  By slightly obfuscating $\varphi$ and $\psi$ (\ie by shifting quantifiers and introducing a unary symbol to get rid of the third variable) we can make our counterexample formulae to be in $\Lsuffix^2$; consult~\cref{appendix:thm:two-variable-fragments-do-not-have-CIP}.
\end{proof}

\subsection{Restoring CIP in $\Lprefix$}\label{subsec:CIP-for-Lprefix}

Even though $\Lsuffix$ and $\Linfix$ fail to have CIP, it turns out that $\Lprefix$ still has it. 
To prove interpolation for $\Lprefix$, we are going to construct a model for two jointly consistent $\Lprefix$ formulae $\varphi(\vartuplex)$ and $\psi(\vartuplex)$. 
However, rather than modifying existing amalgamation-based arguments used, for instance, in~\cite{marx1995algebraic,HooglandMO99,BaranyBC18}, we will construct our model explicitly by specifying prefix-types for tuples.  
We feel that our approach, which is more direct in nature than other arguments found in the literature, could potentially be useful also in other contexts.

Take $\varphi(\vartuplex)$ and $\psi(\vartuplex)$ in normal form~\ref{NForm-Lprefix} satisfying the premise of~\cref{lemma:aux-lemma-for-interpolation}.
Hence, there are structures $\str{A}$ and $\str{B}$ and tuples $\domelemtuplea \in A^k$ and $\domelemtupleb \in B^k$ such that that $\str{A} \models \varphi(\domelemtuplea), \str{B} \models \psi(\domelemtupleb)$ and $(\str{A},\domelemtuplea) \bisimto_{\Lprefix[\sigma]} (\str{B}, \domelemtupleb)$, where $\sigma := \sig(\varphi) \cap \sig(\psi)$. 
Let $\tau := \sig(\varphi) \cup \sig(\psi)$. 

We will define a sequence of $\tau$-structures $\str{U}_1 \leq \ldots \leq \str{U}_M := \str{U}$, where $M = \max \{\arity(\relsymbolR) \mid \relsymbolR \in \tau\}$, satisfying the following inductive assumptions: 
(i) $U_i = \N$, 
(ii) the interpretation of symbols from $\tau$ of arity $> i$ is empty, and 
(iii) for any $i$-tuple $\domelemtuplec$ in $\str{U}_i$ there are $i$-tuples $\domelemtupled$ in $\str{A}$ and $\domelemtuplee$ in $\str{B}$ so that $(\str{A}, \domelemtupled) \bisimto_{\Lprefix[\sigma]} (\str{B}, \domelemtuplee)$ and $\tp{\Lprefix[\tau]}{\str{U}_i}{\domelemtuplec} = \tp{\Lprefix[\sig(\varphi)]}{\str{A}}{\domelemtupled} \cup \tp{\Lprefix[\sig(\psi)]}{\str{B}}{\domelemtuplee}$ hold. 
The last condition guarantees that no tuple $\domelemtuplec$ of $\str{U}_i$ violates the universal requirements of $\varphi$ and $\psi$, since otherwise the corresponding tuple would violate them, contradicting modelhood of $\str{A}$ or~$\str{B}$.\footnote{We note that this claim no longer holds if $\Lprefix$ is replaced by either $\Lsuffix$ or $\Laffix$, which is why the forthcoming construction does not work for these logics.}

For the inductive base, take $\str{U}_1$ with domain $\N$ and empty interpretation of symbols from~$\tau$.
Our goal is to realise each $(\sig(\varphi),1)$-prefix-type, which is realised in $\str{A}$ and $\str{B}$, in~$\str{U}_1$ in a careful way, suggested by the inductive assumption.
Let $t$ be a $(\sig(\varphi),1)$-prefix type realised in $\str{A}$ and let $\domelemc \in A$ be some element witnessing it.
Since $(\str{A},\domelemtuplea) \bisimto_{\Lprefix[\sigma]} (\str{B}, \domelemtupleb)$ holds,
there exists an element $\domelemd$ of $\str{B}$ so that 
$\tp{\Lprefix[\sigma]}{\str{A}}{\domelemc} = \tp{\Lprefix[\sigma]}{\str{B}}{\domelemd}$. 
Now we will assign the $(\tau,1)$-prefix-type 
$\tp{\Lprefix[\sig(\varphi)]}{\str{A}}{\domelemc} \cup \tp{\Lprefix[\sig(\psi)]}{\str{B}}{\domelemd}$ to some element $\domeleme$ of~$\str{U}_1$, for which we have not yet assigned a $(\tau,1)$-prefix-type. 
For the remaining elements of $\str{U}_1$, having no $(\tau,1)$-prefix-type assigned, we assign any of the previously realised types.

Suppose then that $\str{U}_k$ is defined. To define $\str{U}_{k{+}1}$, we will start by providing witnesses for the existential requirements of $\varphi$ and $\psi$;
since the two cases are rather analogous, we will restrict our attention to the former case. 
Consider an existential requirement $\varphi_i^\exists$ of $\varphi(\vartuplex)$ and let $\domelemtuplee \in U_k^k$ be a $k$-tuple so that $\str{U} \models \alpha_i(\domelemtuplee)$. 
By construction, there exists a tuple $\domelemtuplea \in A^k$ witnessing $\tp{\Lprefix[\sig(\varphi)]}{\str{U}_k}{\domelemtuplee} = \tp{\Lprefix[\sig(\varphi)]}{\str{A}}{\domelemtuplea}$. 
Since $\str{A} \models \varphi_i^\exists$, there exists an element $\domelemc\in A$ so that $\str{A} \models \beta_i(\domelemtuplea,\domelemc)$. 
Due to $(\str{A},\domelemtuplea) \bisimto_{\Lprefix[\sigma]} (\str{B}, \domelemtupleb)$, we know that there exists an element $\domelemd \in B$ satisfying $\tp{\Lprefix[\sigma]}{\str{A}}{\domelemtuplea,\domelemc} = \tp{\Lprefix[\sigma]}{\str{B}}{\domelemtupleb,\domelemd}$. 
Now we pick an element $f\in U$ for which we have not yet assigned a $(\tau,k{+}1)$-prefix-type for the tuple $(\domelemtuplee,\domelemf)$ (recall that the domain of our model is $\N$, so such an element always exists). 
We assign the following $(\tau,k{+}1)$-prefix-type to the tuple $(\domelemtuplee,\domelemf)$: $\tp{\Lprefix[\sig(\varphi)]}{\str{A}}{(\domelemtuplea,\domelemc)} \cup \tp{\Lprefix[\sig(\psi)]}{\str{B}}{(\domelemtupleb,\domelemd)}$.
Note that the assigned $(\tau,k{+}1)$-prefix-type is consistent with the $(\tau,k)$-prefix-type that we assigned to $\domelemtuplee$. Having assigned witnesses to relevant existential requirements of $\varphi(\vartuplex)$ and $\psi(\vartuplex)$, there are still $(k{+}1)$-tuples of elements of $\str{U}$ for which we have not yet assigned a $(\tau,k{+}1)$-prefix-type. For those tuples we will assign any $(\tau,k{+}1)$-prefix-type that we have already
assigned to some other $(k{+}1)$-tuple of elements of $\str{U}_{k{+}1}$.
This completes the construction of $\str{U}_{k{+}1}$. 

By construction, it is clear that there exists a tuple $\domelemtuplee$ of elements of $\str{U}$ so that $\domelemtuplee \in \relsymbolH^{\str{U}}$; in particular, $\str{U} \models \varphi(\domelemtuplee) \land \psi(\domelemtuplee)$ holds. 
Thus, by~\cref{lemma:aux-lemma-for-interpolation} we conclude:
\begin{theorem}
$\Lprefix$ enjoys the Craig Interpolation Property.
\end{theorem}

\subsection{Restoring CIP in guarded logics}\label{subsec:CIP-for-Gaffix}

\newcommand\reijosorder{\mathrel{\ooalign{$\prec$\cr
  \hidewidth\raise0.03ex\hbox{$\cdot\mkern0.5mu$}\cr}}}

\newcommand{\lesslex}{<_{\textit{lex}}}

Finally we turn our attention to the logics $\Gaffix$ and present the main contribution of the paper.
It will be convenient to employ suitable tree-like models. Intuitively, HATs~\cite{Bednarczyk21} are just trees in which relations connect elements but only in a level-by-level ascending order; see~\cref{fig:hat}.~HAHs~are~collections~of~HATs.
\begin{definition}\label{def:hat}
    A structure $\str{T}$ is a \emph{higher-arity tree} (HAT) if its domain is a prefix-closed subset of sequences from $\N^*$ and for all relation symbols $\relsymbolR$ we have that $(d_1,\dots,d_k) = \domelemtupled \in \relsymbolR^{\str{T}}$ implies that for each index $i < k$ there exists a number $n_i$ such that $\domelemd_{i+1} = \domelemd_i \cdot n_i$, where $\domelemd_i \cdot n_i$ means that the element $n_i$ is appended to the sequence $\domelemd_i$.
    A structure $\str{H}$ is a \emph{higher-arity hedge} (HAH) if $\str{H}$ becomes a HAT if extended by a single element $\varepsilon$. 
\end{definition}


\begin{figure}[ht] 
    \centering
\begin{tikzpicture}[scale=0.6, transform shape]
  \draw (0,-2) node[ptrond, vert, label=center:\small{0}] (V00) {\phantom{0000}};
  \node[] at (2, -1.8) {\( \relsymbolT \)};

  \draw (-2,-4) node[ptrond, jaune, label=center:\small{00}] (V000) {\phantom{0000}};
  \draw (0,-4) node[ptrond, rouge, label=center:\small{01}] (V001) {\phantom{0000}};
  \draw (-2, -6) node[ptrond, vert, label=center:\small{000}] (V0000) {\phantom{0000}};
  \draw (0, -6) node[ptrond, vert, label=center:\small{010}] (V0010) {\phantom{0000}};

  \path[->] (V00) edge [red] node[xshift=-6] {\(\relsymbolR \)} (V000);
  \path[->] (V00) edge [red] node[xshift=6] {\(\relsymbolR \)} (V001);
  \path[->>,dashdotdotted] (V000) edge [blue] node[xshift=-4] {\(\relsymbolS \)} (V0000);
  \path[->>,dashdotdotted] (V001) edge [blue] node[xshift=4] {\(\relsymbolS \)} (V0010);

  \draw (4, 0) node[ptrond, vert, label=center:\small{$\varepsilon$}] (V0) {\phantom{0000}};
  \path[->>,dashdotdotted] (V0) edge [blue] node[xshift=-10] {\(\relsymbolS \)} (V00);

  \draw (4, -2) node[ptrond, jaune, label=center:\small{1}] (V01) {\phantom{0000}};
  \draw (4, -4) node[ptrond, vert, label=center:\small{10}] (V010) {\phantom{0000}};
  \draw (3, -6) node[ptrond, jaune, label=center:\small{100}] (V0100) {\phantom{0000}};
  \draw (5, -6) node[ptrond, rouge, label=center:\small{101}] (V0101) {\phantom{0000}};

  \path[->] (V0) edge [red] node[xshift=-5] {\(\relsymbolR \)} (V01);
  \path[->>,dashdotdotted] (V01) edge [bend right=30,blue] node[xshift=4] {\(\relsymbolS \)} (V010);
  \path[->] (V01) edge [bend left=30,red] node[xshift=4] {\(\relsymbolR \)} (V010);
  \path[->] (V010) edge [red] node[xshift=4] {\(\relsymbolR \)} (V0101);

  \draw (6, -2) node[ptrond, rouge, label=center:\small{2}] (V02) {\phantom{0000}};
  \path[->>,dashdotdotted] (V0) edge [blue] node[xshift=-10] {\(\relsymbolS \)} (V02);

  \draw (8, -2) node[ptrond, vert, label=center:\small{3}] (V03) {\phantom{0000}};

\scoped[on background layer] \filldraw [blue!10, line width=2.5em, line join=round,] (V01.center) -- (V010.center) -- (V0100.center) -- cycle;
\scoped[on background layer] \filldraw[red!10, line width=1.5em, line join=round,] (V0.center) -- (V00.center) -- (V000.center) -- cycle;

\end{tikzpicture}
\caption{An example HAT $\str{T}$. All relations go down lvl-by-lvl. The red area means~$(\varepsilon,0,00) \in \relsymbolT^{\str{T}}$.}\label{fig:hat}
\end{figure}
By a subtree of a HAT $\str{T}$ rooted at an element $\domelemd$ we mean a substructure of $\str{T}$ with the domain composed of all elements of the form $\domelemd \domelemw$ for a possibly empty word $\domelemw$.
Note that such a subtree is also a HAT after an obvious renaming.

We are going to employ the following lemma, stating that for our purposes we can focus on tree-like models only.
Its proof relies on the suitable notion of unravelling, see~\cref{appendix:sec:proof-unrav}.
\begin{lemma}\label{lemma:hat-vs-jointconsistency}
Let a logic $\logicL$ be any of $\Gaffix$, $\varphi, \psi$ be $\logicL$-formulae and $\sigma := \sig(\varphi) \cap \sig(\psi)$ containing the predicate $\relsymbolH$.
Assume that models $\str{A} \models \varphi(\domelemtuplea), \str{B} \models \psi(\domelemtupleb)$ are given such that $\domelemtuplea \in \relsymbolH^{\str{A}}, \domelemtupleb \in \relsymbolH^{\str{B}}$ and $(\str{A},\domelemtuplea) \bisimto_{\Gaffix[\sigma]} (\str{B}, \domelemtupleb)$ hold.
Then there are HAH models $\str{T}_{\str{A}} \models \varphi(\domelemtuplec), \str{T}_{\str{B}} \models \psi(\domelemtupled)$ satisfying  $\domelemtuplec \in \relsymbolH^{\str{T}_\str{A}}, \domelemtupled \in \relsymbolH^{\str{T}_\str{B}}$ and $(\str{T}_\str{A},\domelemtuplec) \bisimto_{\Gaffix[\sigma]} (\str{T}_\str{B}, \domelemtupled)$.
\end{lemma}

Take $\varphi(\vartuplex)$, $\psi(\vartuplex)$ in form~\ref{NForm-Gaffix} with the same head, satisfying the premise of~\cref{lemma:aux-lemma-for-interpolation}.
We have structures $\str{A}$ and $\str{B}$ and tuples $\domelemtuplea \in A^k$ and $\domelemtupleb \in B^k$ so that $\str{A} \models \varphi(\domelemtuplea), \str{B} \models \psi(\domelemtupleb)$ and $(\str{A},\domelemtuplea) \bisimto_{\Gaffix[\sigma]} (\str{B}, \domelemtupleb)$,
where $\sigma := \sig(\varphi) \cap \sig(\psi)$ and~$\tau := \sig(\varphi) \cup \sig(\psi)$. 
Using~\cref{lemma:hat-vs-jointconsistency} we can assume that $\str{A}$ and $\str{B}$ are $\tau$-HAHs.
As done before, we aim at constructing a $\tau$-structure being a model of both $\varphi$ and $\psi$.
To do this, we will construct a growing sequence of $\tau$-HAHs $\str{U}_0 := \str{A} \leq \str{U}_1 \leq \ldots \leq \str{U}_n \leq \ldots$, whose limit $\str{U}$ will be a model of $\varphi\land\psi$. 
For simplicity, let us employ the following naming scheme.
A tuple $\domelemtupled$ from $\str{U}_{n}$ ($\domelemtupled \sqin U_n$) is called (a) $n$-\emph{fresh} if $\domelemtupled \sqin U_{n{-}1}$ and $n$-\emph{aged} otherwise, (b) \emph{maximal} if its not an affix of any different $\sigma$-live tuple.

A high-level idea of the construction of the sequence $\str{U}_i$, obfuscated by many challenging technical details, is as follows.
Starting from $\str{A}$ we inductively ``complete'' types of all $\sigma$-live tuples to become proper $\tau$-live tuples. 
This will help, if done carefully and in a bisimilarity-preserving way, the structure $\str{U}_i$ to fulfil the universal constraints of $\varphi$ and $\psi$, but may introduce tuples without witnesses for the existential constraints. 
Hence, after each ``completion'' phase, we will ``repair'' the obtained structure by ``copying'' some substructures of $\str{A}$ and $\str{B}$ and ``gluing'' them on existing witness-lacking tuples (providing~the~required~witnesses).

During the construction we will make sure that for every $n$-aged $\sig(\varphi)$-live (resp. $\sig(\psi)$-live) $k$-tuple in $\str{U}$, there exists a $k$-tuple in $\str{A}$ (resp. in~$\str{B}$)
having equal $(\sig(\varphi),k)$-affix-type (resp. $(\sig(\psi),k)$-affix-type).
This will be controlled by means of partial \emph{witness functions} $\wit_{\str{A}}: U_n \to A, \wit_{\str{B}}: U_n \to B$, intuitively pinpointing from where a tuple in $\str{U}$ originated from.
To make the construction work, the witness function will fulfil several technical criteria, that are listed below.
Conditions (\ref{enum:a}) and (\ref{enum:b}) speak about the compatibility of types between a tuple and its witness tuple; this guarantees that no tuple from $\str{U}_n$ violate the universal requirements of $\varphi$ and $\psi$.
Conditions (\ref{enum:c})--(\ref{enum:d}) guarantees the satisfaction of the existential requirements of $\varphi$ and $\psi$ (condition (\ref{enum:c}) takes care of ``local'' requirements while (\ref{enum:d}) handles the ``global'' ones). 
Formally, for every $n$-aged $\domelemtuplec$ from $\str{U}_{n}$ we have that:
\begin{enumerate}[(a)]
   \item \label{enum:a} If $\domelemtuplec$ is $\sigma$-live then both $\domelemtupled := \wit_{\str{A}}(\domelemtuplec)$ and $\domelemtuplee := \wit_{\str{B}}(\domelemtuplec)$ are defined,  $(\str{A},\domelemtupled) \bisimto_{\Gaffix[\sigma]} (\str{B}, \domelemtuplee) \; \text{holds and} \;
   \tp{\Gaffix[\tau]}{\str{U}_n}{\domelemtuplec}$ is equal to $\tp{\Gaffix[\sig(\varphi)]}{\str{A}}{\domelemtupled} \cup \tp{\Gaffix[\sig(\psi)]}{\str{B}}{\domelemtuplee}$.
    
   \item \label{enum:b} If $\domelemtuplec$ is not $\sigma$-live but is $\sig(\varphi)$-live (resp. $\sig(\psi)$-live), then $\domelemtupled := \wit_{\str{A}}(\domelemtuplec)$ (resp. $\domelemtupled := \wit_{\str{B}}(\domelemtuplec)$) is defined, and $\tp{\Gaffix[\sig(\varphi)]}{\str{U}_n}{\domelemtuplec}$ is equal to $\tp{\Gaffix[\sig(\varphi)]}{\str{A}}{\domelemtupled}$ (resp. $\tp{\Gaffix[\sig(\psi)]}{\str{B}}{\domelemtupled}$).
    
   \item \label{enum:c} if $\domelemtuplec$ is $\sig(\varphi)$-live (resp. $\sig(\psi)$-live), then for every existential requirement $\lambda := \relsymbolR_i( \vartuplexfromto{1}{\ell_i}) \to \exists \vartuplexfromto{\ell_i}{\ell_i {+} k_i} (\relsymbolS_i(\vartuplexfromto{1}{\ell_i {+} k_i}) \land \theta_i(\vartuplexfromto{1}{\ell_i {+} k_i}))$ from $\varphi$ (resp. from $\psi$) with $\domelemtuplec$ satisfying the premise of $\lambda$, there is a tuple $\domelemtupled$ in $\str{U}_{n}$ so that $\domelemtuplec\domelemtupled$ satisfies the conclusion of $\lambda$.

   \item \label{enum:d} For every $\sig(\varphi)$-live (resp. $\sig(\psi)$-live) tuple $\domelemtupled$ from $\str{A}$ (resp. from $\str{B}$) there is a tuple $\domelemtuplee$ in $\str{U}_1$ such that $\tp{\Gaffix[\sig(\varphi)]}{\str{U}_n}{\domelemtuplee} = \tp{\Gaffix[\sig(\varphi)]}{\str{A}}{\domelemtupled}$ (resp. $\tp{\Gaffix[\sig(\varphi)]}{\str{U}_n}{\domelemtuplee} = \tp{\Gaffix[\sig(\psi)]}{\str{B}}{\domelemtupled}$).

\noindent \hspace{-2em} While the following property is not necessary to guarantee that the limit $\str{U}$ is a model of
$\varphi \land \psi$, it plays an important technical role in the construction:
   \item \label{enum:e} If $\domelemtupled$ is an $n$-fresh $\sigma$-live tuple such that either $\wit_{\str{A}}(\domelemtupled)$ or $\wit_{\str{B}}(\domelemtupled)$ is undefined, then for every prefix $\domelemtupledfromto{1}{k}$ of $\domelemtupled$ that is contained in $U_{n{-}1}$, meaning that $\domelemtupledfromto{1}{k} \sqin U_{n{-}1}$, there exists an $n$-aged $\sigma$-live tuple $\domelemtuplec$ which contains $\domelemtupledfromto{1}{k}$ as its affix.
\end{enumerate}
Using conditions (\ref{enum:a})--(\ref{enum:d}) it follows that $\str{U} \models \varphi \land \psi$, allowing us to conclude (by~\cref{lemma:aux-lemma-for-interpolation}):
\begin{theorem}\label{thm:Gaffix-have-CIP}
$\Ginfix, \Gsuffix$ and $\Gprefix$ enjoy the Craig Interpolation Property.
\end{theorem}
We will now move on to the construction of $\str{U}$, described below.
We start from the crucial, aforementioned notions of \emph{completions} and \emph{repairs}.
Intuitively the \emph{completion} just ``completes a type of a tuple'' in a bisimulation-preserving way, taking all symbols of $\tau$ into account. 
\emph{Repair} simply ``plugs in'' certain subtrees from $\str{A}$ or $\str{B}$ into $\str{U}$, providing missing witnesses.
\begin{definition}[completion]
Let $(\str{T}, \domelemtupled)$ be a pointed $\tau$-HAH, where $\domelemtupled$ is $\sigma$-live with $\wit_{\str{A}}, \wit_{\str{B}}$ defined.
The $\domelemtupled$-\emph{completion} of $\str{T}$ is obtained from $\str{T}$ by redefining interpretation of symbols from $\tau$ in a min. way so that $\tp{\Gaffix[\tau]}{\str{T}}{\domelemtupled}$ equals $\tp{\Gaffix[\sig(\varphi)]}{\str{A}}{\wit_{\str{A}}(\domelemtupled)} \cup \tp{\Gaffix[\sig(\psi)]}{\str{B}}{\wit_{\str{B}}(\domelemtupled)}$.
\end{definition}

\begin{definition}[repair]
Let $(\str{T}, \domelemtuplec)$ be a pointed $\tau$-HAH with only $\domelemtupled := \wit_{\str{A}}(\domelemtuplec)$ defined, where $\domelemtuplec$ is $\sigma$-live in $\str{T}$.
Suppose also that there is a tuple $\domelemtuplee$ in $\str{B}$ such that $(\str{A},\domelemtupled) \bisimto_{\Gaffix[\sigma]} (\str{B}, \domelemtuplee)$ holds.
The $(\str{B}, \domelemtuplee)$-\emph{repair} of $\domelemtuplec$ is a $\tau$-HAH $\str{T}'$ obtained from $\str{T}$ in the following five steps:
\begin{enumerate}
    \item Let $\str{B}_0$ be the subtree of $\str{B}$ rooted at the first element of $\domelemtuplee$. 
   \item Take $\str{T}'$ to be the union of $\str{T}$ and $\str{B}_0$ without $\domelemtuplee$.
    \item $\str{T}'$ will contain $\str{T}$ as a substructure.
    \item By identifying $\domelemtuplec$ with $\domelemtuplee$, we interpret the relation symbols for tuples of elements of $\str{T}' \upharpoonright B_0$ in such a way that
            the resulting substructure of $\str{T}'$ is isomorphic with $\str{B}_0$.
    \item We set $\wit_{\str{B}}$ on freshly added elements to be the identity on $\str{B}_0$.
\end{enumerate}
The substructure $\str{T}' \upharpoonright (B_0 \cup \domelemtuplec)$ is called a $\domelemtuplec$-\emph{component} of $\str{T}'$. $\str{T}'$ becomes a HAH after a routine renaming.
We define $(\str{A}, \domelemtupled)$-\emph{repair} of $\domelemtuplec$ analogously.
\end{definition}

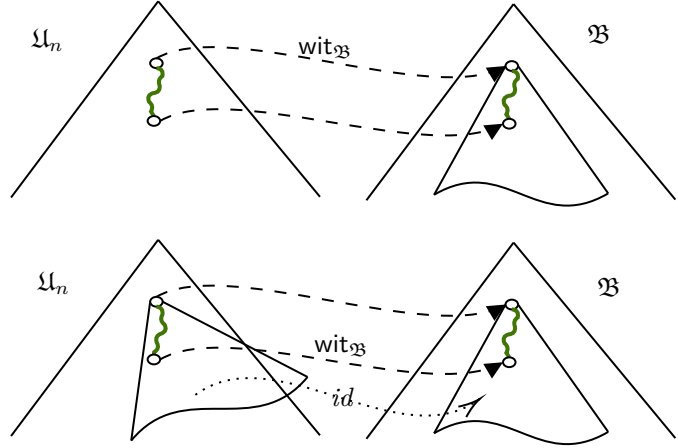
\begin{figure}[H]
\centering

\tikzset{every picture/.style={line width=0.75pt}} 

\begin{tikzpicture}[x=0.75pt,y=0.75pt,yscale=-1,xscale=1]

\draw    (80.45,4) -- (161.21,102.35) ;
\draw    (80.45,4) -- (7,102.7) ;
\draw   (76.44,35.16) .. controls (76.44,33.81) and (77.84,32.71) .. (79.56,32.71) .. controls (81.28,32.71) and (82.68,33.81) .. (82.68,35.16) .. controls (82.68,36.51) and (81.28,37.61) .. (79.56,37.61) .. controls (77.84,37.61) and (76.44,36.51) .. (76.44,35.16) -- cycle ;
\draw [color={rgb, 255:red, 65; green, 117; blue, 5 }  ,draw opacity=1 ][line width=1.5]    (79.56,37.61) .. controls (82.07,37.8) and (83.05,39.15) .. (82.52,41.67) .. controls (80.85,42.76) and (80.4,44.31) .. (81.17,46.32) .. controls (81.44,48.65) and (80.4,49.93) .. (78.07,50.16) .. controls (75.68,50.81) and (74.96,52.29) .. (75.92,54.61) .. controls (77.59,55.94) and (77.9,57.56) .. (76.84,59.47) -- (78.31,61.73) ;
\draw   (75.19,64.18) .. controls (75.19,62.83) and (76.59,61.73) .. (78.31,61.73) .. controls (80.03,61.73) and (81.43,62.83) .. (81.43,64.18) .. controls (81.43,65.54) and (80.03,66.64) .. (78.31,66.64) .. controls (76.59,66.64) and (75.19,65.54) .. (75.19,64.18) -- cycle ;
\draw    (257.84,5.4) -- (338.6,103.75) ;
\draw    (257.84,5.4) -- (184.39,104.1) ;
\draw   (253.83,36.56) .. controls (253.83,35.21) and (255.22,34.11) .. (256.95,34.11) .. controls (258.67,34.11) and (260.07,35.21) .. (260.07,36.56) .. controls (260.07,37.91) and (258.67,39.01) .. (256.95,39.01) .. controls (255.22,39.01) and (253.83,37.91) .. (253.83,36.56) -- cycle ;
\draw [color={rgb, 255:red, 65; green, 117; blue, 5 }  ,draw opacity=1 ][line width=1.5]    (256.95,39.01) .. controls (259.46,39.2) and (260.44,40.55) .. (259.9,43.07) .. controls (258.23,44.16) and (257.78,45.71) .. (258.55,47.72) .. controls (258.82,50.05) and (257.79,51.33) .. (255.45,51.56) .. controls (253.06,52.21) and (252.35,53.69) .. (253.31,56.01) .. controls (254.98,57.34) and (255.29,58.96) .. (254.23,60.87) -- (255.7,63.13) ;
\draw   (252.58,65.58) .. controls (252.58,64.23) and (253.98,63.13) .. (255.7,63.13) .. controls (257.42,63.13) and (258.82,64.23) .. (258.82,65.58) .. controls (258.82,66.94) and (257.42,68.04) .. (255.7,68.04) .. controls (253.98,68.04) and (252.58,66.94) .. (252.58,65.58) -- cycle ;
\draw  [dash pattern={on 4.5pt off 4.5pt}]  (81.43,64.18) .. controls (116.37,43.6) and (212.96,84.88) .. (250.36,66.77) ;
\draw [shift={(252.58,65.58)}, rotate = 149.5] [fill={rgb, 255:red, 0; green, 0; blue, 0 }  ][line width=0.08]  [draw opacity=0] (8.93,-4.29) -- (0,0) -- (8.93,4.29) -- cycle    ;
\draw  [dash pattern={on 4.5pt off 4.5pt}]  (79.56,32.71) .. controls (114.5,12.12) and (214.08,55.76) .. (251.61,37.74) ;
\draw [shift={(253.83,36.56)}, rotate = 149.5] [fill={rgb, 255:red, 0; green, 0; blue, 0 }  ][line width=0.08]  [draw opacity=0] (8.93,-4.29) -- (0,0) -- (8.93,4.29) -- cycle    ;
\draw    (260.07,36.56) -- (305.08,100.63) ;
\draw    (253.83,36.56) -- (218.17,101.33) ;
\draw    (218.17,101.33) .. controls (253.83,80.32) and (269.43,121.64) .. (305.08,100.63) ;

\draw (16,16.4) node [anchor=north west][inner sep=0.75pt]    {$\str{U}_n$};
\draw (293,14.4) node [anchor=north west][inner sep=0.75pt]    {$\str{B}$};
\draw (149,20.4) node [anchor=north west][inner sep=0.75pt]    {$\wit_{\str{B}}$};
\end{tikzpicture}

\begin{tikzpicture}[x=0.75pt,y=0.75pt,yscale=-1,xscale=1]

\draw    (82.45,4) -- (163.21,102.35) ;
\draw    (82.45,4) -- (9,102.7) ;
\draw   (78.44,35.16) .. controls (78.44,33.81) and (79.84,32.71) .. (81.56,32.71) .. controls (83.28,32.71) and (84.68,33.81) .. (84.68,35.16) .. controls (84.68,36.51) and (83.28,37.61) .. (81.56,37.61) .. controls (79.84,37.61) and (78.44,36.51) .. (78.44,35.16) -- cycle ;
\draw [color={rgb, 255:red, 65; green, 117; blue, 5 }  ,draw opacity=1 ][line width=1.5]    (81.56,37.61) .. controls (83.97,37.91) and (84.96,39.28) .. (84.53,41.71) .. controls (83.48,43.62) and (83.82,45.24) .. (85.55,46.56) .. controls (86.73,48.71) and (86.21,50.25) .. (83.98,51.18) .. controls (81.73,51.93) and (80.99,53.44) .. (81.76,55.73) .. controls (82.83,57.72) and (82.39,59.31) .. (80.44,60.48) -- (80.31,61.73) ;
\draw   (77.19,64.18) .. controls (77.19,62.83) and (78.59,61.73) .. (80.31,61.73) .. controls (82.03,61.73) and (83.43,62.83) .. (83.43,64.18) .. controls (83.43,65.54) and (82.03,66.64) .. (80.31,66.64) .. controls (78.59,66.64) and (77.19,65.54) .. (77.19,64.18) -- cycle ;
\draw    (259.84,5.4) -- (340.6,103.75) ;
\draw    (259.84,5.4) -- (186.39,104.1) ;
\draw   (255.83,36.56) .. controls (255.83,35.21) and (257.22,34.11) .. (258.95,34.11) .. controls (260.67,34.11) and (262.07,35.21) .. (262.07,36.56) .. controls (262.07,37.91) and (260.67,39.01) .. (258.95,39.01) .. controls (257.22,39.01) and (255.83,37.91) .. (255.83,36.56) -- cycle ;
\draw [color={rgb, 255:red, 65; green, 117; blue, 5 }  ,draw opacity=1 ][line width=1.5]    (258.95,39.01) .. controls (261.46,39.2) and (262.44,40.55) .. (261.9,43.07) .. controls (260.23,44.16) and (259.78,45.71) .. (260.55,47.72) .. controls (260.82,50.05) and (259.79,51.33) .. (257.45,51.56) .. controls (255.06,52.21) and (254.35,53.69) .. (255.31,56.01) .. controls (256.98,57.34) and (257.29,58.96) .. (256.23,60.87) -- (257.7,63.13) ;
\draw   (254.58,65.58) .. controls (254.58,64.23) and (255.98,63.13) .. (257.7,63.13) .. controls (259.42,63.13) and (260.82,64.23) .. (260.82,65.58) .. controls (260.82,66.94) and (259.42,68.04) .. (257.7,68.04) .. controls (255.98,68.04) and (254.58,66.94) .. (254.58,65.58) -- cycle ;
\draw  [dash pattern={on 4.5pt off 4.5pt}]  (83.43,64.18) .. controls (118.37,43.6) and (214.96,84.88) .. (252.36,66.77) ;
\draw [shift={(254.58,65.58)}, rotate = 149.5] [fill={rgb, 255:red, 0; green, 0; blue, 0 }  ][line width=0.08]  [draw opacity=0] (8.93,-4.29) -- (0,0) -- (8.93,4.29) -- cycle    ;
\draw  [dash pattern={on 4.5pt off 4.5pt}]  (81.56,32.71) .. controls (116.5,12.12) and (216.08,55.76) .. (253.61,37.74) ;
\draw [shift={(255.83,36.56)}, rotate = 149.5] [fill={rgb, 255:red, 0; green, 0; blue, 0 }  ][line width=0.08]  [draw opacity=0] (8.93,-4.29) -- (0,0) -- (8.93,4.29) -- cycle    ;
\draw    (262.07,36.56) -- (307.08,100.63) ;
\draw    (255.83,36.56) -- (220.17,101.33) ;
\draw    (220.17,101.33) .. controls (255.83,80.32) and (271.43,121.64) .. (307.08,100.63) ;
\draw    (84.68,35.16) -- (157,73) ;
\draw    (78.44,35.16) -- (69,105) ;
\draw    (69,105) .. controls (95.21,72.97) and (130.79,105.03) .. (157,73) ;
\draw  [dash pattern={on 0.84pt off 2.51pt}]  (101,82) .. controls (140.6,52.3) and (201.76,112.77) .. (241.79,84.87) ;
\draw [shift={(243,84)}, rotate = 143.13] [color={rgb, 255:red, 0; green, 0; blue, 0 }  ][line width=0.75]    (10.93,-3.29) .. controls (6.95,-1.4) and (3.31,-0.3) .. (0,0) .. controls (3.31,0.3) and (6.95,1.4) .. (10.93,3.29)   ;

\draw (159,50) node [anchor=north west][inner sep=0.75pt]    {$\wit_{\str{B}}$};
\draw (167,77) node [anchor=north west][inner sep=0.75pt]    {$\textit{id}$};
\draw (21,18.4) node [anchor=north west][inner sep=0.75pt]    {$\str{U}_n$};
\draw (300,21.4) node [anchor=north west][inner sep=0.75pt]    {$\str{B}$};
\end{tikzpicture}

\caption{An example structure $\str{U}_n$ before and after we performed a ``$\str{B}$''-repair.}
\end{figure}

We proceed with the base of induction, setting first $\str{U}_0$ to be $\str{A}$. It will be four-fold.\\

\noindent\textbf{Base case: Step I.} We set up $\wit_{\str{A}}$ and $\wit_{\str{B}}$ functions. 
For $\wit_{\str{A}}$ we will simply take the identity function.
To define $\wit_{\str{B}}$, we intuitively proceed by traversing $\str{U}_0$ from top to bottom. 
More precisely, let $L_0^{\str{A}}$ denote the set of all maximal $\sigma$-live tuples in $\str{U}_0$. 
Letting $\lesslex$ denote the lexicographic ordering of $\N^*$, we construct a well-founded linear ordering $\reijosorder$ on $L_0^{\str{A}}$ as follows: $\domelemtuplec \reijosorder \domelemtupled$ iff there is an $i \leq \min\{|\domelemtuplec|,|\domelemtupled|\}$ such that $\domelemc_i \lesslex \domelemd_i$ and $\domelemc_j = \domelemd_j$ for every $j \lesslex i$ (note that if there is no such $i$, then the tuples are equal due to maximality).
One can show that $\domelemtuplec \reijosorder \domelemtupled$ implies that $(\heartsuit){:}$ if $\domelemtuplec$ and $\domelemtupled$ share some elements, then there exists $i,j$ and $k$ such that $\domelemtuplecfromto{i}{j} = \domelemtupledfromto{1}{k}$ and none of the elements $\domelemd_\ell$, for $\ell > k$, occur in $\domelemtuplec$.
To prove this, one needs to simply show that if $\domelemd_k$ occurs in~$\domelemtuplec$, then $(\domelemd_1,\ldots,\domelemd_k)$ is an affix of $\domelemtuplec$ (the proof goes via careful inspection of the definition of HAHs, \cf~\cref{appendix:HAHs-are-important-xD}).

We define $\wit_{\str{B}}$ inductively w.r.t $\reijosorder$. 
Consider a maximal $\sigma$-live tuple $\domelemtupled$ and suppose that we have defined $\wit_{\str{B}}$ for all the $\sigma$-live tuples $\domelemtuplec \reijosorder \domelemtupled$. 
There are two cases to consider.
\begin{itemize}
   \item There exists a tuple $\domelemtuplec \reijosorder \domelemtupled$ sharing at least one element with $\domelemtupled$. 
   By $(\heartsuit)$, for every such tuple $\domelemtuplec$ there are $i,j$ and $k$ so that $\domelemtuplecfromto{i}{j} = \domelemtupledfromto{1}{k}$ and none of the elements $d_\ell$, for $\ell > k$, occur in $\domelemtuplec$. 
   Let $\domelemtuplec$ be the tuple for which the corresponding value $k$ is the largest. 
   Since $\domelemtuplec$ is $\sigma$-live, by induction hypothesis there exists some $\domelemtuplee \sqin B$ such that $(\str{A},\domelemtuplec) \bisimto_{\Gaffix[\sigma]} (\str{B}, \domelemtuplee)$ holds. 
   Thus there exists some $\domelemtuplef \sqin B$ such that $\domelemtupleffromto{i}{j} = \domelemtupleefromto{i}{j}$ and $(\str{A},\domelemtupled) \bisimto_{\Gaffix[\sigma]} (\str{B}, \domelemtuplef)$. 
   We now extend $\wit_{\str{B}}$ in such a way that $\wit_{\str{B}}(\domelemtupled) = \domelemtuplef$.

   \item Otherwise $\domelemtuplec$ and $\domelemtupled$ do not share any elements. 
   Since $\domelemtupled$ is $\sigma$-live and $(\str{A},\domelemtuplea) \bisimto_{\Gaffix[\sigma]} (\str{B}, \domelemtupleb)$, there exists some $\domelemtuplee \sqin B$ such that $(\str{A},\domelemtupled) \bisimto_{\Gaffix[\sigma]} (\str{B}, \domelemtuplee)$. 
   We then simply extend $\wit_{\str{B}}$ in such a way that $\wit_{\str{B}}(\domelemtupled) = \domelemtuplee$.\\
   \hspace{-2em}The resulting mapping $\wit_{\str{B}}$. This finishes Step I.\\
\end{itemize}
\noindent\textbf{Base case: Step II.} We next complete types of all fresh (= all in this case) $\sigma$-live tuples of $\str{U}_0$. 
Take any maximal $\sigma$-live tuple $\domelemtupled$ from $\str{U}_0$ and perform the $\domelemtupled$-completion of $\str{U}_0$. 
It is easy to see that this process is conflict-free in the following sense: there is no tuple $\domelemtuplec$ and $R\in \sig(\psi)$ so that we end up specifying both $\domelemtuplec \in \relsymbolR^{\str{U}_0}$ and $\domelemtuplec \not\in \relsymbolR^{\str{U}_0}$. 
First, if two maximal $\sigma$-live tuples $\domelemtupled$ and $\domelemtuplee$ have affixes $\domelemtupledfromto{i}{j}$ and $\domelemtupleefromto{k}{\ell}$ such that $\domelemtupledfromto{i}{j}=\domelemtupleefromto{k}{\ell}$, then we know that $\wit_\str{B}(\domelemtupledfromto{i}{j})=\wit_\str{B}(\domelemtupleefromto{k}{\ell})$, and thus there are no conflicts in the ``intersections'' of $\sigma$-live tuples.  
Second, by construction $\tp{\Gaffix[\sigma]}{\str{A}}{\domelemtupled}=\tp{\Gaffix[\sigma]}{\str{B}}{\wit_\str{B}(\domelemtupled)}$ holds for all maximal $\sigma$-live tuple $\domelemtupled$, and hence the $\sigma$-infix-types that we assigned to $\sigma$-live tuples are indeed types, \ie they are consistent.
Thus our process is conflict-free. 

Note that our structure satisfies now conditions (\ref{enum:a}) and (\ref{enum:b}).\\

\noindent\textbf{Base case: Step III.} We finish the base case by providing witnesses for fresh $\sig(\psi)$-tuples via repairs.
Recall that $L_0^{\str{A}}$ denotes the set of all maximal $\sigma$-live tuples in $\str{U}_0^{\str{A}}$. 
For each $\domelemtupled \in L_0^{\str{A}}$ we perform the $(\str{B}, \wit_{\str{B}}(\domelemtupled))$-repair of $\domelemtupled$; the resulting structure will be taken to be $\str{U}_1$. 
Note that now every $\sig(\psi)$-live tuple in $\str{U}_1$ has its witnesses for the existential requirements, but there may be new $\sig(\varphi)$-live tuples without them. 
Moreover, $\wit_{\str{B}}$ is defined for all freshly added elements, but $\wit_{\str{A}}$ is not. 
Furthermore, we note that $\str{U}_1$ now satisfies condition (\ref{enum:e}), since all the $1$-fresh live tuples for which $\wit_{\str{A}}$ is not defined are present in the subtrees that we attached to $\str{U}_0$ during the repair, which is done only at (maximal) $\sigma$-live tuples.\\

\noindent\textbf{Base case: Step IV.} It could be the case that the structure $\str{U}_1$ produced in the previous step violates (\ref{enum:d}), due to the lack of realisation of a certain type from $\str{B}$. 
Thus, as an extra precaution, unique to the base case, we add a disjoint copy of $\str{B}$ to $\str{U}_1$ and define $\wit_{\str{B}}$ for it to be the identity. 
Note that now (\ref{enum:d}) will be satisfied in any extension of $\str{U}_1$.\\

\noindent\textbf{Inductive step.} 
The inductive step is analogous to Steps I-III from the base case, hence we keep its description short.
Assume that $\str{U}_n$ is defined and that in the previous step of the construction we employed $\str{B}$-repairs (the case of $\str{A}$-repairs is symmetric).
Given a component $\str{C}$ that was created during such a repair, we let $L_n^\str{C}$ denote the set of all maximal $n$-fresh $\sigma$-live tuples in $\str{C}$.
Since $\str{C}$ is essentially a HAT (up to renaming), we can again define a well-founded linear order $\reijosorder$ on $L_n^\str{C}$ in the same way as we did in the base case for $L_0^{\str{A}}$.
As in the base case, we then define missing values of $\wit_{\str{A}}$ for elements of $\str{C}$ inductively w.r.t~$\reijosorder$.

Observe that some of the tuples in $L_n^\str{C}$ might contain a proper prefix of elements of $U_{n{-}1}$. 
In the case of $\Gsuffix$ these tuples do not cause any problems to us, because suffix-types do not impose any constraints on proper prefixes. 
In the cases of $\Gprefix$ and $\Ginfix$ we handle these tuples by using the fact that $\str{U}_n$ satisfies condition (\ref{enum:e}) as follows. 
Let $\domelemtupled \in L_n^\str{C}$ be such a tuple and let $k$ be the largest index so that $\domelemtupledfromto{1}{k} \sqin U_{n{-}1}$. 
Using condition (\ref{enum:e}), we know that there exists a $\sigma$-live $n$-aged tuple $\domelemtuplec$ so that $\domelemtuplecfromto{i}{j} = \domelemtupledfromto{1}{k}$, for some $i$ and $j$. 
Employing condition (\ref{enum:a}), we know that $\domelemtuplee := \wit_{\str{A}}(\domelemtuplec)$, $\domelemtuplef := \wit_{\str{B}}(\domelemtuplec)$ and $\domelemtupleh := \wit_{\str{B}}(\domelemtupled)$ are defined and that $(\str{A},\domelemtuplee) \bisimto_{\Gaffix[\sigma]} (\str{B}, \domelemtuplef)$. 
Thus there exists a $\sigma$-live tuple $\domelemtupleg \sqin A$ such that $\domelemtupleffromto{i}{j} = \domelemtuplegfromto{i}{j}$ and $(\str{A},\domelemtupleg) \bisimto_{\Gaffix[\sigma]} (\str{B}, \domelemtupleh)$. 
We now extend $\wit_{\str{A}}$ in such a way that $\wit_{\str{A}}(\domelemtupled) = \domelemtupleg$.
 
The above procedure is repeated for all components $\str{C}$ that were introduced during the previous repair. 
Having defined $\wit_{\str{A}}$ for all the elements, we perform a completion that works for exactly the same reasons as described before.
Finally, letting $L_n$ denote the set of all maximal $n$-fresh $\sigma$-live tuples, we perform repair of every tuple in $L_n$, which results in a model that we select as $\str{U}_{n{+}1}$. 
We stress that every $\sig(\varphi)$-live tuple in $\str{U}_{n{+}1}$ has its witnesses for the existential requirements, but there can now be new $\sig(\psi)$-live tuples without them. 
We also emphasise that $\wit_{\str{A}}$ is defined for all the new elements but $\wit_{\str{B}}$ might not be. 
This concludes the inductive step and hence, also the construction of $\str{U}$ and the proof of~\cref{thm:Gaffix-have-CIP}.\\

We conclude the construction with the following remark.

\begin{remark}
The presented model construction is quite generic. 
Indeed, the only part of the construction which is really specific to $\Gaffix$ is the first step of the construction, namely the part where we define inductively the values of witness functions. 
We expect that the presented technique can be easily adapted to other logics, especially to other fragments of the guarded fragments. 
For instance, we believe that our technique can be adjusted, \eg to the case of the two-variable $\GF$ from~\cite{HooglandMO99} as well as to the uniform one-dimensional $\GF$ from~\cite{Jaakkola22}.
\end{remark}


\section{Conclusions}
In this paper kick-started a project of understanding the model theory of the family of guarded and unguarded ordered logics. 
We first investigated the relative expressive power of ordered logics by means of suitable bisimulations.
Afterwards, we proceed with the Craig Interpolation Property (CIP) showing that 
(i) the fluted and the forward fragments do not enjoy CIP,
(ii) while the other logics that we consider enjoy it. 
The fact that the fluted fragment does not posses CIP was quite unexpected in the light of already existing claims for the contrary~\cite[Thm. 14]{Purdy02}.
For the other logics we proposed a novel model-theoretic ``complete-and-repair'' method of creating a model out of two bisimilar forest-like structures. 

There are several interesting future work directions.

\begin{enumerate}

\item One example is to investigate the Łoś-Tarski Preservation Theorem as well as other preservation theorems. 
While we think that we already have a working construction for guarded ordered logics, the status of ŁTPT holding for $\Lprefix, \Lsuffix$, and $\Linfix$ is not clear.\footnote{Purdy provides a ``proof'' in~\cite{Purdy02} that $\Lsuffix$ has ŁTPT. However, his ``proof'' is sketchy and lacks sufficient mathematical arguments required to verify its correctness. In the light of our discovery of yet another false claim from~\cite{Purdy02}, we believe that it is safe to assume that ŁTPT for $\Lsuffix$ \emph{is open}.}

\item Another work direction is to take a look at on effective interpolation, similarly to what has been proposed in~\cite{BenediktCB16} as well as on the interpolant existence problem for $\Linfix$ and $\Lsuffix$, as done in~\cite{JungW21}. Preliminary results were obtained. 
It is also interesting whether the guarded ordered logics enjoy stronger versions of interpolations, \eg Lyndon's interpolation or Otto's interpolation. We are quite optimistic about it.

\item Can our upper bounds for the model checking problem for ordered logics be lifted to the case of the list encoding of structures?
\end{enumerate}

We are also actively working on the finitary versions of van Benthem theorem for the forward guarded fragment as well as the Lindström-style characterisation theorems.
This is an ongoing work of Benno Fünfstück, a master student at TU Dresden, under the supervision of B. Bednarczyk.

\clearpage
\bibliography{references}
\clearpage
\appendix
\tableofcontents

\section{Appendix for~\cref{subsec:model-checking}}

\subsection{Combined complexity of $\Lprefix$ and $\Lsuffix$}\label{appendix:model-checking-fluted-logic}

We will first present an alternating algorithm for the model checking problem of $\Lsuffix$, which requires only logarithmic amount of work-space, after which we will argue that a similar algorithm works also for $\Lprefix$. Since alternating $\LogSpace$ equals $\PTime$, the desired results will follow from this. Let $\varphi$ be a formula of either $\Lprefix$ or $\Lsuffix$.  Let $N = \max\{ar(R) \mid R \in \sig(\varphi)\}$. Our algorithm will make crucial use of the following fact: if $|\str{A}|$ denotes the size of the encoding of the structure $\str{A}$, then $\log_2(|\str{A}|) \geq \log_2(|A|^N) = N\log_2(|A|)$. 

\begin{algorithm}[h]
  \DontPrintSemicolon
  \KwData{An $\Lsuffix$ formula $\varphi(\varx_m,\dots,\varx_n)$, a structure $\str{A}$ \\ and an assignment $s:\{\varx_m,\dots,\varx_n\} \to A$}
  \caption{\textbf{ModelCheck}$(\str{A},s,\varphi)$}\label{algo:lsuffix-model-checking}

  \textbf{if} $\varphi(\varx_m,\dots,\varx_n) = R(\varx_m,\dots,\varx_n)$ and $(s(\varx_m),\dots,s(\varx_n)) \in R^{\str{A}}$ \textbf{then} \\      
  \ \textbf{return} \texttt{True}, otherwise \textbf{return} \texttt{False}

   \textbf{if} $\varphi(\varx_m,\dots,\varx_n)   = \psi(\varx_m,\dots,\varx_n) \lor \psi'(\varx_m,\dots,\varx_n)$ \textbf{then}\\
	\ \textbf{guess} $i\in \{1,2\}$ and \textbf{return Modelcheck}$(\str{A},s,\psi_i)$

  \textbf{if} $\varphi(\varx_m,\dots,\varx_n)   = \psi(\varx_m,\dots,\varx_n) \land \psi'(\varx_m,\dots,\varx_n)$ \textbf{then}\\
	\ \textbf{choose} $i\in \{1,2\}$ and \textbf{return Modelcheck}$(\str{A},s,\psi_i)$

  \textbf{if} $\varphi(\varx_m,\dots,\varx_n) = \exists \varx_{n+1} \psi(\varx_m,\dots,\varx_n,\varx_{n+1})$ \textbf{then} \\
  \ \textbf{if} $n - m + 1 \geq N$ \textbf{then} $s = s - \{(x_m,s(x_m))\}$ \\
  \ \textbf{guess} $a\in A$ and \textbf{return ModelCheck}$(\str{A},s(a/x_{n+1}),\psi)$

  \textbf{if} $\varphi(\varx_m,\dots,\varx_n) = \forall \varx_{n+1} \psi(\varx_m,\dots,\varx_n,\varx_{n+1})$ \textbf{then} \\
  \ \textbf{if} $n - m + 1 \geq N$ \textbf{then} $s = s - \{(x_m,s(x_m))\}$ \\
  \ \textbf{choose} $a\in A$ and \textbf{return ModelCheck}$(\str{A},s(a/x_{n+1}),\psi)$
\end{algorithm}

\cref{algo:lsuffix-model-checking} describes (informally) an alternating model checking algorithm for $\Lsuffix$. The algorithm is not really specific to the logic $\Lsuffix$, but the fact that it requires only logarithmic amount of work-space is. Indeed, at each step the algorithm will store a pointer to the current subformula and an assignment to at most $N$ variables (since previous variables are removed from the assignment when the total number of variables in the domain of the current assignment is about to become more than $N$), which in total requires only $O(\log_2(|\varphi|) + N\log_2(|A|))$ bits of memory. 

Notice that the algorithm indeed works for $\Lsuffix$, since the input formula has the following syntactical properties: its width is at most $N$ and if it has a subformula $\psi(x_m,\dots,x_n)$, where $(x_m,\dots,x_n)$ lists precisely the free variables of $\psi$, then $\psi$ does not have a subformula that has any of the variables $\{x_1,\dots,x_m\}$ as its free variables. On the other hand, $\Lprefix$ has similar syntactical properties: its width is at most $N$ and on every maximal path in the syntactical tree of any $\Lprefix$ formula of width at most $N$ at most $N$ distinct variables can occur. Thus it should be obvious that an algorithm very similar to~\cref{algo:lsuffix-model-checking} can be designed for $\Lprefix$.

\section{Appendix for~\cref{sec:expressive-power}}

\subsection{Proof of~\cref{lemma:linking-bisimilarity-and-equivalence}}\label{appendix:lemma:linking-bisimilarity-and-equivalence}


\begin{proof}
	Here we will restrict our attention to the case of $\Laffix$, since the case of $\Gaffix$ can be proved exactly the same way. We start with the left to right direction. Suppose that $\str{A}, \domelemtuplea \bisimto_{\Laffix[\sigma]} \str{B}, \domelemtupleb$ and let $\bisimZ$ be the promised bisimulation between  $(\omegasat{\str{A}},\domelemtuplec)$ and $(\omegasat{\str{B}},\domelemtupled)$.
	We claim that for every $n\in \N$, $\varphi(\vartuplexfromto{i}{j}) \in \Laffix(n)[\sigma]$ and $(\domelemtuplec,\domelemtupled) \in \bisimZ$, where $|\domelemtuplec|=|\domelemtupled|=n$, we have that
	\[\str{A} \models \varphi(\domelemtuplecfromto{i}{j}) \iff \str{B} \models \varphi(\domelemtupledfromto{i}{j}).\]
	Since $(\domelemtuplea,\domelemtupleb) \in \bisimZ$, this would imply that $\str{A}, \domelemtuplea \equiv_{\Laffix[\sigma]} \str{B}, \domelemtupleb$. To prove this claim, we use induction over the structure of $\Laffix[\sigma]$. Here we will
	limit ourselves on the case where $\varphi(\vartuplexfromto{i}{j}) = \exists \varx_{j+1} \psi(\vartuplexfromto{i}{j+1})$. Pick $(\domelemtuplec,\domelemtupled) \in \bisimZ$ and suppose that $\str{A} \models  \varphi(\domelemtuplecfromto{i}{j})$. Thus
	there exists $\domeleme\in A$ so that $\str{A} \models \psi(\domelemtuplecfromto{i}{j}\domeleme)$. Since $(\domelemtuplec,\domelemtupled) \in \bisimZ$, by \ref{bisim:forth} we know that there exists $\domelemf\in B$ so that $(\domelemtuplecfromto{i}{j}\domeleme,\domelemtupledfromto{i}{j}\domelemf) \in \bisimZ$.
	By induction, we have that $\str{B} \models \psi(\domelemtupledfromto{i}{j}f)$ and hence $\str{B} \models \varphi(\domelemtupledfromto{i}{j})$, which is what we wanted to show. The converse direction can be proved analogously using~\ref{bisim:back}.

	We then move on to the right to left direction. Suppose that $\str{A}$ and $\str{B}$ are $\omega$-saturated structures for which $\str{A}, \domelemtuplea \equiv_{\Laffix[\sigma]} \str{B}, \domelemtupleb$. We claim that the following set
	\[\bisimZ := \bigg\{(\domelemtuplec,\domelemtupled) \in \bigcup_{n < \omega} (A^n \times B^n) \mid \str{A}, \domelemtuplec \equiv_{\Laffix[\sigma]} \str{B}, \domelemtupled\bigg\}\]
	is a $\Laffix[\sigma]$-bisimulation between $(\str{A},\domelemtuplea)$ and $(\str{B},\domelemtupleb)$. First, by construction we know that $(\domelemtuplea,\domelemtupleb) \in \bisimZ$. Furthermore it is clear that $\bisimZ$ satisfies \ref{bisim:atomic-harmony}.

	Towards showing that $\bisimZ$ satisfies \ref{bisim:forth}, pick $(\domelemtuplec,\domelemtupled) \in \bisimZ$, an affix $\domelemtuplecfromto{i}{j}$ of $\domelemtuplec$ and an element $\domeleme \in A$. We claim that the following set
	\[\Sigma(\varx_{j + 1}) := \{\varphi(\domelemtuplecfromto{i}{j},\varx_{j + 1}) \mid \str{A} \models \varphi(\domelemtuplecfromto{i}{j},e), \varphi(\vartuplexfromto{i}{j {+} 1}) \in \Laffix[\sigma]\}\]
	where we use $\varphi(\vartuplexfromto{i}{j + 1})$ to denote a formula whose free variables form an affix of the sequence $\vartuplexfromto{i}{j + 1}$, is realised in $\str{B}$. Since $\str{B}$ is $\omega$-saturated, it suffices to show that each finite subset of
	$\Sigma(\varx_{j + 1})$ is realised in $\str{B}$. So, let $\Sigma'(\varx_{j + 1}) \subseteq \Sigma(\varx_{j + 1})$ be a finite set. By definition, we know that $\str{A} \models \exists \varx_{j+1} \bigwedge \Sigma'(\varx_{j+1})$. Since $\str{A}, \domelemtuplec \equiv_{\Laffix[\sigma]} \str{B}, \domelemtupled$,
	we can deduce that $\str{B} \models \exists \varx_{j+1} \bigwedge \Sigma'(\varx_{j+1})$, and hence $\Sigma'(\varx_{j+1})$ is indeed realised in $\str{B}$.

	Thus we know that $\Sigma(\varx_{j{+}1})$ is realised in $\str{B}$, say by $\domelemf$. Due to the choice of $\Sigma(\varx_{j+1})$, we know that $\str{A}, \domelemtuplecfromto{i}{j}\domeleme \equiv_{\Laffix[\sigma]} \str{B}, \domelemtupledfromto{i}{j}\domelemf$, and hence that
	$(\domelemtuplecfromto{i}{j}\domeleme,\domelemtupledfromto{i}{j}\domelemf) \in \bisimZ$. Thus $\bisimZ$ satisfies \ref{bisim:forth}. The proof that it satisfies \ref{bisim:back} is entirely analogous.
\end{proof}




\subsection{Proof of~\cref{thm:expressive-power-full-characterisation}}\label{appendix:thm:expressive-power-full-characterisation}

  We first show that $\Linfix$ and $\Lsuffix$ are equally expressive. 
  This is done by a routine rewriting process.

  \begin{lemma}\label{lemma:Linfix-and-Lsuffix-have-the-same-expressive-power}
    $\Linfix \approx \Lsuffix$. 
    Moreover, for a given $\varphi \in \Lsuffix$, the equivalent formula in $\Lsuffix$ can be computed in $|\varphi|$-fold exponential time. 
  \end{lemma}
  \begin{proof}
    We have $\Lsuffix \preceq \Linfix$ by definition. 
    To show $\Linfix \preceq \Lsuffix$, take any $\varphi \in \Linfix$ that is not in $\Lsuffix$.
    W.l.o.g. we assume that $\varphi$ does not contain any proper subsentences.
    Let $\mathcal{Q}{\vary}\; \psi(\vartuplex,\vary)$ be a maximally nested subformula of $\varphi$ violating the definition of $\Lsuffix$.
    Next, turn $\psi$ into the disjunctive normal form (treating subformulae of $\psi$ starting from quantifiers as atomic ones). 
    Finally, put outside $\psi$ these disjuncts of $\psi$ in which $\vary$ does not appear. 
    The obtained formula is obviously equivalent to $\varphi$ and the modified subformulae $\mathcal{Q}{\vary} \psi(\vartuplex,\vary)$ is in $\Lsuffix$.
    Hence, by repeating the process we eventually reach the formula in $\Lsuffix$ that is equivalent to $\varphi$. 
    The second part of the lemma follows directly from the proof (turning a formula into DNF requires exponential time and pessimistically we must convert every subformulae of $\varphi$ into DNF). 
  \end{proof}

  By chasing the presented diagram, it is easy to see that in order to conclude the proofs it suffices to show that the formulae $\forall{\varx_1}{\varx_2}{\varx_3}\; (\relsymbolR(\varx_1,\varx_2,\varx_3) \to \relsymbolT(\varx_2,\varx_3))$ is not definable in $\Lprefix$, and that $\forall{\varx_1}{\varx_2}{\varx_3}\; (\relsymbolR(\varx_1,\varx_2,\varx_3) \to \relsymbolS(\varx_1,\varx_2))$ is not definable in $\Gsuffix$. 
  We prove it in the consecutive lemmas.
  Below we depict the structures $\str{A}$ and $\str{B}$ that we will use in the proofs.

  
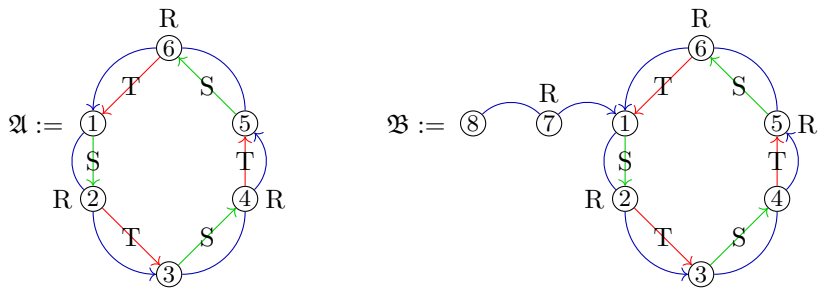
\begin{figure}[H]
    \centering
    \begin{tikzpicture}[transform shape]
        \draw (-0.75, 0) node[] (A) {$\str{A} :=$};

        \draw (0, 0) node[ptrond, label=center:\small{$1$}] (OneA) {};
        \draw (0, -1) node[ptrond, label=center:\small{$2$}] (TwoA) {};
        \draw (1,-2) node[ptrond, label=center:\small{$3$}] (ThreeA) {};
        \draw (2, -1) node[ptrond, label=center:\small{$4$}] (FourA) {};
        \draw (2, 0) node[ptrond, label=center:\small{$5$}] (FiveA) {};
        \draw (1, 1) node[ptrond, label=center:\small{$6$}] (SixA) {};

        \path[->] (OneA) edge[vert] node[] {\( \relsymbolS \)} (TwoA); 
        \path[->] (TwoA) edge[rouge] node[] {\( \relsymbolT \)} (ThreeA); 
        \path[->] (ThreeA) edge[vert] node[] {\( \relsymbolS \)} (FourA); 
        \path[->] (FourA) edge[rouge] node[] {\( \relsymbolT \)} (FiveA); 
        \path[->] (FiveA) edge[vert] node[] {\( \relsymbolS \)} (SixA); 
        \path[->] (SixA) edge[rouge] node[] {\( \relsymbolT \)} (OneA); 

        \path[-] (OneA) edge[bend right=45, color=bleu] node[] {} (TwoA);
        \path[->] (TwoA) edge[bend right=45, color=bleu] node[] {} (ThreeA);
        \draw (-0.4, -1) node[] {$\relsymbolR$};

        \path[-] (ThreeA) edge[bend right=45, color=bleu] node[] {} (FourA);
        \path[->] (FourA) edge[bend right=45, color=bleu] node[] {} (FiveA);
        \draw (2.4, -1) node[] {$\relsymbolR$};

        \path[-] (FiveA) edge[bend right=45, color=bleu] node[] {} (SixA);
        \path[->] (SixA) edge[bend right=45, color=bleu] node[] {} (OneA);
        \draw (1, 1.4) node[] {$\relsymbolR$};

        \draw (4.25, 0) node[] (B) {$\str{B} :=$};

        \draw (7, 0) node[ptrond, label=center:\small{$1$}] (OneB) {};
        \draw (7, -1) node[ptrond, label=center:\small{$2$}] (TwoB) {};
        \draw (8,-2) node[ptrond, label=center:\small{$3$}] (ThreeB) {};
        \draw (9, -1) node[ptrond, label=center:\small{$4$}] (FourB) {};
        \draw (9, 0) node[ptrond, label=center:\small{$5$}] (FiveB) {};
        \draw (8, 1) node[ptrond, label=center:\small{$6$}] (SixB) {};
        \draw (6, 0) node[ptrond, label=center:\small{$7$}] (SevenB) {};
        \draw (5, 0) node[ptrond, label=center:\small{$8$}] (EightB) {};

        \path[->] (OneB) edge[vert] node[] {\( \relsymbolS \)} (TwoB); 
        \path[->] (TwoB) edge[rouge] node[] {\( \relsymbolT \)} (ThreeB); 
        \path[->] (ThreeB) edge[vert] node[] {\( \relsymbolS \)} (FourB); 
        \path[->] (FourB) edge[rouge] node[] {\( \relsymbolT \)} (FiveB); 
        \path[->] (FiveB) edge[vert] node[] {\( \relsymbolS \)} (SixB); 
        \path[->] (SixB) edge[rouge] node[] {\( \relsymbolT \)} (OneB); 

        \path[-] (EightB) edge[bend left=45, color=bleu] node[] {} (SevenB);
        \path[->] (SevenB) edge[bend left=45, color=bleu] node[] {} (OneB);
        \draw (6, 0.4) node[] {$\relsymbolR$};

        \path[-] (OneB) edge[bend right=45, color=bleu] node[] {} (TwoB);
        \path[->] (TwoB) edge[bend right=45, color=bleu] node[] {} (ThreeB);
        \draw (6.6, -1) node[] {$\relsymbolR$};

        \path[-] (ThreeB) edge[bend right=45, color=bleu] node[] {} (FourB);
        \path[->] (FourB) edge[bend right=45, color=bleu] node[] {} (FiveB);
        \draw (9.4, 0) node[] {$\relsymbolR$};

        \path[-] (FiveB) edge[bend right=45, color=bleu] node[] {} (SixB);
        \path[->] (SixB) edge[bend right=45, color=bleu] node[] {} (OneB);
        \draw (8, 1.4) node[] {$\relsymbolR$};

    \end{tikzpicture}
  \end{figure}

  \begin{lemma}
    $\forall{\varx_1}{\varx_2}{\varx_3}\; (\relsymbolR(\varx_1,\varx_2,\varx_3) \to \relsymbolT(\varx_2,\varx_3))$ cannot be expressed in $\Lprefix$.
  \end{lemma}
  \begin{proof}
    Consider $\{ \relsymbolR, \relsymbolT \}$-structures $\str{A}, \str{B}$ defined as:
    \begin{itemize}\itemsep0em
      \item $A = \{ 1,2,3,4,5,6 \}$, $\relsymbolT^{\str{A}} = \{ (2,3), (4,5), (6,1) \}$, $\relsymbolR^{\str{A}} = \{ (1,2,3), (3,4,5), (5,6,1) \}$,
      \item $B = A \cup \{ 7,8 \}$, $\relsymbolT^{\str{B}} = \relsymbolT^{\str{A}}$ and $\relsymbolR^{\str{B}} = \relsymbolR^{\str{A}} \cup \{ (8,7,1) \}$.
    \end{itemize}
    Note that $\str{A}$ satisfies the required property (call it $\varphi$) but $\str{B}$ does not, since $(8,7,1) \in \relsymbolR^{\str{B}}$ but $(7,1) \not\in \relsymbolT^{\str{B}}$.
    Let $f : B \to A$, be identity on $A$, and let $f(7)=1, f(8)=2$ hold.
    We claim that the set $\bisimZ$, defined as $\bisimZ = \bigcup_{n < \omega} \bisimZ_n$, where:
    \begin{itemize}\itemsep0em
      \item $\bisimZ_0 = \{ (\epsilon, \epsilon) \}$,
      \item $\bisimZ_{n{+}1} = \bisimZ_n \cup \{ (\domelemtuplea f(\domelemd), \domelemtuplea \domelemd) \mid \domelemd \in B \}$, 
    \end{itemize}
    is a $\Lprefix[\{ \relsymbolR, \relsymbolT \}]$-bisimulation between $\str{A}$ and $\str{B}$. This implies, by~\cref{lemma:linking-bisimilarity-and-equivalence}, that $\varphi$ is not definable in $\Lprefix[\{ \relsymbolR, \relsymbolT \}]$.

    Note that \ref{bisim:forth} and \ref{bisim:back} are satisfied by the definition of $\bisimZ_{i{+}1}$.
    Checking \ref{bisim:atomic-harmony} boils down to verifying by hand all the possible cases, but there are not some many of them since the $\Laffix\{ \relsymbolR, \relsymbolT \}$-types of tuples longer than $3$ are equal to the $\Laffix\{ \relsymbolR, \relsymbolT \}$-type of their $3$-element prefixes.
  \end{proof}

  \begin{lemma}
  $\forall{\varx_1}{\varx_2}{\varx_3}\; (\relsymbolR(\varx_1,\varx_2,\varx_3) \to \relsymbolS(\varx_1,\varx_2))$ cannot be expressed in $\Gsuffix$.
  \end{lemma}
  \begin{proof}
    Similarly to the above lemma, consider $\{ \relsymbolR, \relsymbolS \}$-structures $\str{A}, \str{B}$ defined as:
      \begin{itemize}\itemsep0em
        \item $A = \{ 1,2,3,4,5,6 \}$, $\relsymbolS^{\str{A}} = \{ (1,2), (3,4), (5,6) \}$, $\relsymbolR^{\str{A}} = \{ (1,2,3), (3,4,5), (5,6,1) \}$,
        \item $B = A \cup \{ 7,8 \}$, $\relsymbolS^{\str{B}} = \relsymbolS^{\str{A}}$ and $\relsymbolR^{\str{B}} = \relsymbolR^{\str{A}} \cup \{ (8,7,1) \}$.
      \end{itemize}
    Note that $\str{A}$ satisfies the required property (call it $\psi$) but $\str{B}$ does not, since $(8,7,1) \in \relsymbolR^{\str{B}}$ but $(8,7) \not\in \relsymbolS^{\str{B}}$.
    We show that $\str{A}$ and $\str{B}$ are $\Gsuffix[\{ \relsymbolR, \relsymbolS \}]$-bisimilar.
    Again, it would imply that $\psi$ is not $\Gsuffix[\{ \relsymbolS, \relsymbolR \}]$-definable.

    Let $f:B \to A$ be identity on $A$, and let $f(7) = 1, f(8) = 2$ hold. We define a set $\bisimZ$ by specifying that for every live tuple $\domelemtupleb$ of elements of $B$ the pair $(f[\domelemtupleb],\domelemtupleb)$ is added to $\bisimZ$. We claim that $\bisimZ$ is a $\Gaffix[\{R,S\}]$-bisimulation between $\str{A}$ and $\str{B}$. First, it is straightforward to check that it satisfies \ref{bisim:atomic-harmony}, since there is only a small number of live tuples in $\str{B}$. The fact that $\bisimZ$ satisfies \ref{bisim:back} follows immediately from its construction. Verifying that $\bisimZ$ satisfies \ref{bisim:forth} requires checking a small number of cases and we will omit it here.
\end{proof}

This finishes the exhaustive study of relative expressive power among guarded and unguarded ordered logics.

\section{Appendix for~\cref{sec:interpolation}}

\subsection{Proof of~\cref{lemma:joint-consistency-vs-interpolation}}\label{appendix:lemma:joint-consistency-vs-interpolation}
  \begin{proof}
    For the proof we basically follow the proof of Lemma 2 from the appendix of~\cite{JungW21}.

    Suppose there is an interpolant $\chi \in \logicL[\tau]$ with~$\varphi \models \chi \models \psi$, but $\varphi$ and $\neg \psi$ are jointly $\logicL[\tau]$-consistent.
    By joint-consistency, take $\str{A} \models \varphi$ and $\str{B} \models \neg \psi$.
    Then $\str{A} \models \chi$ ($\chi$ is an interpolant), $\str{B} \models \chi$ (by joint $\logicL[\tau]$-consistency), and hence, $\str{B} \models \psi$ ($\chi$ is an interpolant). 
    A contradiction.

    For the other direction, assume that for all $\str{A}, \str{B}$ satisfying $\str{A} \models \varphi$ and $\str{B} \models \neg\psi$ we have $\str{A} \not\bisimto_{\logicL[\tau]} \str{B}$. 
    Now take $\Phi := \{ \lambda \in \logicL[\tau] \mid \varphi \models \lambda \}$ to be the set of $\logicL[\tau]$-consequences of $\varphi$. 
    Obviously $\varphi \models \Phi$. 
    We will also show that $\Phi \models \psi$. 
    Let $\str{B}$ be a model of $\Phi$ and let $\Omega := \{ \chi \in \logicL[\tau] \mid \str{B} \models \chi \} \cup \{ \varphi \}$. 
    By compactness and the definition of $\Omega$, we conclude that $\Omega$ is satisfiable and let $\str{A}$ be its model. 
    Now, we pass to $\omega$-saturated extensions $\hat{\str{A}}, \hat{\str{B}}$ of $\str{A}, \str{B}$.
    Note that $\hat{\str{A}} \equiv_{\logicL[\tau]} \hat{\str{B}}$, and hence we conclude $\hat{\str{A}} \bisimto_{\logicL[\tau]} \hat{\str{B}}$ (by~\cref{lemma:linking-bisimilarity-and-equivalence}).
    By our initial assumption we conclude $\str{B} \models \psi$.
    Hence, we proved that $\Phi \models  \psi$.
    By compactness there is a finite subset $\Phi_0$ of $\Phi$ that entails $\psi$, and $\bigwedge \Phi_0$ is the desired interpolant for $\phi \models \psi$.
  \end{proof}


\subsection{Proof of~\cref{lemma:normal-forms}}\label{appendix:sec:normal-forms}

Let $\varphi(\vartuplex), \psi(\vartuplex) \in \logicL$, where $\logicL \in \{\Lprefix, \Gaffix\}$. 
Suppose that there exists structures $\str{A}$ and $\str{B}$ such that $\str{A} \models \varphi(\domelemtuplea), \str{B} \models \psi(\domelemtupleb)$ and $(\str{A},\domelemtuplea) \bisimto_{\logicL[\tau]} (\str{B},\domelemtupleb)$. 
We start by replacing $\varphi(\vartuplex)$ and $\psi(\vartuplex)$ with the formulae $\relsymbolH(\vartuplex) \land \forall \vartuplex (\relsymbolH(\vartuplex) \to \varphi(\vartuplex))$ and $\relsymbolH(\vartuplex) \land \forall \vartuplex (\relsymbolH(\vartuplex) \to \psi(\vartuplex))$ respectively. 
Letting $\varphi'(\vartuplex)$ and $\psi'(\vartuplex)$ denote the resulting formulae, we clearly have that they entail, respectively, $\varphi(\vartuplex)$ and $\psi(\vartuplex)$.
Take $\str{A}'$ and $\str{B}'$ to be the extensions of $\str{A}$ and $\str{B}$ obtained by setting $\relsymbolH^{\str{A}'} = \{\domelemtuplea\}$ and $\relsymbolH^{\str{B}'} = \{\domelemtupleb\}$.
Hence, by bisimilarity of $(\str{A}, \domelemtuplea)$ and $(\str{B}, \domelemtupleb)$ and the way we interpret $\relsymbolH$, we infer $(\str{A}', \domelemtuplea) \bisimto_{\logicL[\tau \cup \{\relsymbolH\}]} (\str{B}', \domelemtupleb)$. 
What remains to be done is to convert the subsentences $\forall \vartuplex (\relsymbolH(\vartuplex) \to \varphi(\vartuplex))$ and $\forall \vartuplex (\relsymbolH(\vartuplex) \to \psi(\vartuplex))$ of $\varphi'(\vartuplex)$ and $\psi'(\vartuplex)$, respectively, into normal form, and to appropriately extend the models $\str{A}'$ and $\str{B}'$.

Let $\varphi := \forall \vartuplex (\relsymbolH(\vartuplex) \to \varphi(\vartuplex))$ and $\psi := \forall \vartuplex (\relsymbolH(\vartuplex) \to \psi(\vartuplex))$. 
We start with the case where $\varphi, \psi \in \Gaffix$. 
We repeat the following rewriting process on $\varphi$: take any maximally nested subformula $\chi$ of $\varphi$, with the maximal block of quantifiers, having the form
\[
        \chi := \exists{\vartuplexfromto{\ell}{k}} \; (\relsymbolS(\vartuplexfromto{r}{k}) \land \lambda(\vartuplexfromto{r}{k})),
\]
where $r \leq \ell$ and $\lambda \in \Gaffix$ is quantifier-free. 
We first discuss the case when $\chi$ is a sentence. 
We then check whether $\str{A} \models \chi$: if so, then we substitute $\chi$ in $\varphi$ with $\top$ and append $\chi$ as a conjunct to the resulting sentence in normal form.
Otherwise, $\chi$ is substituted by $\bot$ and the appended formula will be $\neg\chi$.

Next, we discuss the case when $\chi$ is not a sentence. 
Since $\varphi$ is a sentence, $\chi$ must occur in a scope of an another formula of the form
\[
	\chi' := \exists{\vartupleyfromto{\ell'}{k'}} \; (\relsymbolS'(\vartupleyfromto{s}{k'}) \land \lambda'(\vartupleyfromto{s}{k'})),
\]
where $s \leq \ell'$ and $\vartupleyfromto{s}{k'}$ contains $\vartuplexfromto{r}{\ell {-} 1}$ as an affix. 
Letting $\relsymbolS'$ denote the guard of the innermost such formula, we replace $\varphi$ with the following sentence
\[
	\varphi[\chi(\vartuplexfromto{r}{\ell {-} 1})/\relsymbolR(\vartuplexfromto{r}{\ell {-} 1})] \land \forall \vartuplexfromto{1}{\ell {-} r {-} 1}(\relsymbolR(\vartuplexfromto{1}{\ell {-} r {-} 1}) \to \exists \vartuplexfromto{\ell {-} r}{k {-} r}(\relsymbolS(\vartuplexfromto{1}{k{-}r}) \land \lambda (\vartuplexfromto{1}{k-r})))
\]
\[
	\land \forall \vartuplexfromto{1}{k' {-} s}(\relsymbolS'(\vartuplexfromto{1}{k' {-} s}) \to (\neg \relsymbolR(\vartuplexfromto{r {-} s}{\ell {-} s {-} 1}) \to \forall \vartuplexfromto{\ell {-} s}{k {-} s} (\relsymbolS(\vartuplexfromto{r {-} s}{k {-} s}) \to \neg \lambda(\vartuplexfromto{r {-} s}{k {-} s})))),
\]
where $\relsymbolR$ is a fresh relation symbol and $\varphi[\chi(\vartuplexfromto{r}{\ell {-} 1})/\relsymbolR(\vartuplexfromto{r}{\ell {-} 1})]$ is the sentence obtained from $\varphi$ by replacing the subformula $\chi(\vartuplexfromto{r}{\ell {-} 1})$ with the atom $\relsymbolR(\vartuplexfromto{r}{\ell {-} 1})$.

The above rewriting is applied until we arrive at a sentence $\varphi'$ which is in normal form~\ref{NForm-Gaffix}.
It is easy to see that $\varphi' \models \varphi$ and that there exists an extension $\str{A}''$ of $\str{A}'$ which is a model of $\varphi'$. 
Similarly, by repeating the above rewriting on $\psi$, we arrive at a sentence $\psi'$ in normal form~\ref{NForm-Gaffix} with the property that $\psi' \models \psi$. 
As in the case of $\varphi'$, there exists an extension $\str{B}''$ of $\str{B}'$ which is a model of $\psi'$. 
Note that we were always using fresh auxiliary relation symbols $\relsymbolR$ during the rewriting processes. 
Hence, this guarantees that  $\sig(\varphi') \cap \sig(\psi') = \tau \cup \{ \relsymbolH \}$ and thus also $(\str{A}'', \domelemtuplea) \bisimto_{\logicL[\tau \cup \{\relsymbolH\}]} (\str{B}'', \domelemtupleb)$.

The proof in the case of $\Lprefix$ is mostly analogous to the case of $\Gaffix$, the main difference being that we will be replacing subformulae of the form $\chi := \exists \varx_k \lambda(\vartuplexfromto{1}{k})$
with fresh atomic formulae. 
Otherwise the proof in the case of $\Gaffix$ goes through also in the case of $\Lprefix$.

\subsection{Proof of~\cref{lemma:aux-lemma-for-interpolation}}\label{appendix:lemma:aux-lemma-for-interpolation}
\begin{proof}
  Suppose that there are $\logicL$-formulae $\varphi(\vartuplex), \psi(\vartuplex)$ such that $\varphi(\vartuplex) \models \psi(\vartuplex)$, but there is no interpolant for this entailment. 
  By~\cref{lemma:joint-consistency-vs-interpolation}, $\varphi(\vartuplex)$ and $\neg \psi(\vartuplex)$ are jointly consistent.
  Using~\cref{lemma:normal-forms}, we can deduce that there are two $\logicL$-formulae $\varphi'(\vartuplex)$ and $\psi'(\vartuplex)$ which are jointly consistent and which entail $\varphi(\vartuplex)$ and
  $\neg \psi(\overline{x})$ respectively. 
  Hence, we infer that there is a model $\str{U} \models \varphi'(\vartuplex) \land \psi'(\vartuplex)$.
  Thus $\varphi(\vartuplex) \land \neg \psi(\vartuplex)$ is satisfiable, contradicting $\varphi(\vartuplex) \models \psi(\vartuplex)$.
\end{proof}


\newcommand{\Gaifmanof}[1]{\str{G}_{#1}}
\newcommand{\pathrho}{\rho}
\newcommand{\last}{\mathsf{last}}
\newcommand{\PathsGaifmanof}[1]{\mathsf{Paths}_{\Gaifmanof{#1}}}
\newcommand{\vertexu}{\domelem{u}}
\newcommand{\vertexv}{\domelem{v}}
\newcommand{\vertexw}{\domelem{w}}
\newcommand{\unrav}[1]{{#1}^{\leadsto}}

\subsection{Proof of~\cref{lemma:hat-vs-jointconsistency}}\label{appendix:sec:proof-unrav}
A bit of preliminaries first.

Given a tuple $\domelemtupled$ let $\last(\domelemtupled)$ denote the last element of $\domelemtupled$.
Moreover, given a tuple of tuples/sequences $\domelemtupled$ we denote with $\last[\domelemtupled]$ the tuple obtained by ``mapping'' the last function to every element of $\domelemtupled$, namely the tuple $(\last(\domelemtupled_1), \last(\domelemtupled_2), \ldots)$.

The \emph{forward Gaifman graph} $\Gaifmanof{\str{A}} = (V_{\Gaifmanof{\str{A}}}, E_{\Gaifmanof{\str{A}}})$ of a structure $\str{A}$ is a directed graph with the set of nodes equal to $A$ and with directed edges between $\domelemd$ and $\domeleme$ whenever there are (possibly empty) tuples of elements $\domelemtuplec, \domelemtuplef$ and a relational symbol $\relsymbolR$ witnessing $(\domelemtuplec, \domelemd, \domeleme, \domelemtuplef) \in \relsymbolR^{\str{A}}$.
A~path $\pathrho := \vertexv_1 \vertexv_2 \ldots \vertexv_n$ in the forward Gaifman graph $\Gaifmanof{\str{A}}$ of $\str{A}$ is a finite non-empty word from $(V_{\Gaifmanof{\str{A}}})^+$ in which two consecutive positions $\vertexv_i, \vertexv_{i{+}1}$ are connected by a directed edge in~$\Gaifmanof{\str{A}}$. 
The set of all paths in $\Gaifmanof{\str{A}}$ is denoted by $\PathsGaifmanof{\str{A}}$.

We now can proceed with a suitable notion of unravelling:

\begin{definition}\label{def:unravelling}
    A \emph{HAH-unravelling} of a structure $\str{A}$ is a structure $\unrav{\str{A}}$ with 
    the domain $\PathsGaifmanof{\str{A}}$ and with interpretation of relations defined as follows: for each $k$-ary symbol $\relsymbolR$ we put $\domelemtuplecfromto{1}{k} \in \relsymbolR^{\unrav{\str{A}}}$ iff $(\last(\domelemtuplec_1), \last(\domelemtuplec_2), \ldots, \last(\domelemtuplec_k)) \in \relsymbolR^{\str{A}}$.
\end{definition}

We next observe that the HAH-unravellings of countable structures produce HAHs. 
\begin{lemma}\label{lemma:unravelling-makes-HATs}
Given a countable $\str{A}$ with the domain $A \subseteq \N$, the structure $\unrav{\str{A}}$ is a HAH.\@
\end{lemma}
\begin{proof}
    Note that the domain of $\unrav{\str{A}}$ is a prefix closed (since the set of all paths is closed under taking subpaths) subset of $\N^+$ (since $A \subseteq \N$). 
    Moreover, the requirement on interpreting relations (namely the condition from~\cref{def:hat}) is fulfilled by the way how we define relations in the unravelled structure.
\end{proof}

A crucial property of unravellings is their preservation of infix-types, as stated below.
\begin{lemma}\label{lemma:unravelling-preserves-types}
    Let $(\str{A}, \domelema)$ be a pointed $\tau$-structure, and be its unravelling $\str{B} := \unrav{\str{A}}$.
    Then for all $\sigma$-live tuples $\domelemtupled \sqin B$ we have $\tp{\Laffix[\sigma]}{\str{A}}{\last[\domelemtupled]} = \tp{\Laffix[\sigma]}{\str{B}}{\domelemtupled}$.
\end{lemma}
    \begin{proof}
    Let $\domelemtuplee$ be an infix of $\domelemtupled$. 
    Since $\domelemtupled$ is $\sigma$-live it satisfies the condition on interpretation of relations from~\cref{def:unravelling}, thus so does $\domelemtuplee$.
    By definition of HAH it follows that $\domelemtuplee \in \relsymbolS^{\str{B}}$ iff $\last[\domelemtuplee] \in \relsymbolS^{\str{A}}$. 
    By the choice of $\domelemtuplee$ we infer the desired equality $\tp{\Laffix[\sigma]}{\str{A}}{\last[\domelemtupled]} = \tp{\Laffix[\sigma]}{\str{B}}{\domelemtupled}$.
\end{proof}

Now we can show that HAH-unravellings are bisimilarity preserving.
\begin{lemma}\label{lemma:aux-hat-bisimilation-between-restricted-structures}
For every pointed $\sigma$-structure $(\str{A}, \domelemtuplea)$ with a $\sigma$-live $\domelemtuplea$ we have that $\str{A} \bisimto_{\Ginfix[\sigma]} (\unrav{\str{A}}, \domelemtupleb)$, where $\domelemtupleb$ is the unique $|\domelemtuplea|$-tuple satisfying $\domelemtupleb_i = \domelemtupleafromto{1}{i}$ for all $1 \leq i \leq |\domelemtuplea|$.
\end{lemma}
\begin{proof}
By~\cref{lemma:unravelling-preserves-types} it is immediate to check that the set $\bisimZ$ composed of all pairs of the form $(\domelemtupled, \last[\domelemtupled])$ is a $\Ginfix$-bisimulation between $\unrav{\str{A}}$ and $\str{A}$.
\end{proof}

What remains to be done is to show~\cref{lemma:hat-vs-jointconsistency}.
\begin{proof}
Take $\str{A} \models \varphi(\domelemtuplea)$ and $\str{B} \models \psi(\domelemtupleb)$ such that $(\str{A},\domelemtuplea) \bisimto_{\Gaffix[\sigma]} (\str{B}, \domelemtupleb)$.
Let $\str{C}, \str{D}$ be HAT-unravellings, respectively, of $\str{A}$ and $\str{B}$. 
Take a $|\domelemtuplea|$-tuple $\domelemtuplec$ (resp. $\domelemtupled$) defined as $\domelemtuplec_i := \domelemtupleafromto{1}{i}$ (resp. $\domelemtupled_i := \domelemtuplebfromto{1}{i}$) for all $1 \leq i \leq |\domelemtuplea|$. Note that $\domelemtuplec \sqin C$ (resp. $\domelemtuplec \sqin D$) since $\domelemtuplea \in \relsymbolH^{\str{A}}$ (resp. $\domelemtupleb \in \relsymbolH^{\str{B}}$).
By~\cref{lemma:aux-hat-bisimilation-between-restricted-structures} we infer $(\str{A}, \domelemtuplea) \bisimto_{\Gaffix[\sig(\varphi)]} (\str{C}, \domelemtuplec)$ and $(\str{B}, \domelemtupleb) \bisimto_{\Gaffix[\sig(\psi)]} (\str{D}, \domelemtupled)$.

To complete the proof it suffices to show the following two properties:
\begin{itemize}
    \item $(\str{C}, \domelemtuplec) \models \varphi$ and $(\str{D}, \domelemtupled) \models \psi$.\\
    Follows from~\cref{lemma:linking-bisimilarity-and-equivalence}.

    \item  $(\str{C}, \domelemtuplec) \bisimto_{\Gaffix[\sigma]} (\str{D}, \domelemtupled)$.\\
    As $\bisimto_{\Gaffix[\sigma]} \supseteq \bisimto_{\Gaffix[\sig(\varphi)]}$ and $\bisimto_{\Gaffix[\sigma]} \supseteq \bisimto_{\Gaffix[\sig(\psi)]}$ we infer $(\str{A}, \domelemtuplea) \bisimto_{\Gaffix[\sigma]} (\str{C}, \domelemtuplec)$ and $(\str{B}, \domelemtupleb) \bisimto_{\Gaffix[\sigma]} (\str{D}, \domelemtupled)$. 
    Now, by applying symmetry and transitivity of $\bisimto$ we are done.
\end{itemize}
This completes the proof. 
\end{proof}

\subsection{Missing details from the proof of~\cref{thm:FL-and-FF-doesnt-have-CIP}}\label{appendix:thm:FL-and-FF-doesnt-have-CIP} 
  Recall that 
  \[ 
  \varphi :=  \forall{\varx_1}\forall{\varx_2}\forall{\varx_3}[(\relsymbolR(\varx_1,\varx_2) \land \relsymbolR(\varx_2,\varx_3)) \to (\relsymbolP_1(\varx_1) \land \relsymbolP_2(\varx_3))] \land \; \forall{\varx_1}\forall{\varx_2}[(\relsymbolP_1(\varx_1) \land \relsymbolP_2(\varx_2)) \to \relsymbolR(\varx_1,\varx_2)],
  \]
  \[
  \psi := \exists{\varx_1}\exists{\varx_2} \exists{\varx_3} [\relsymbolR(\varx_1,\varx_2) \land \relsymbolR(\varx_2,\varx_3) \land \relsymbolQ_1(\varx_1) \land \relsymbolQ_2(\varx_3)] \land \;  \forall{\varx_1}\forall{\varx_2}[(\relsymbolQ_1(\varx_1) \land \relsymbolQ_2(\varx_2)) \to \neg\relsymbolR(\varx_1,\varx_2)],
  \]
  and let us define $(\sig(\varphi) \cap \sig(\psi)) = \{ \relsymbolR \}$-structures $\str{A}, \str{B}$ with the domains $A=\{\domelema, \domelemb, \domelemc\}, B=\{1,2,3\}$ and $\relsymbolR^{\str{A}} =  \{(\domelema,\domelemb),(\domelemb,\domelemc),(\domelema,\domelemc),(\domelemc,\domelemc)\}$ and $\relsymbolR^{\str{B}} =  \{(1,2),(2,3),(3,3)\}$, as depicted below:

  
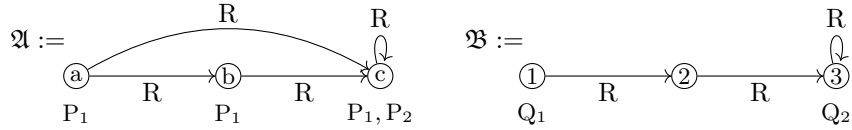
\begin{figure}[H]
    \centering
    \begin{tikzpicture}[transform shape]
        \node[] at (1.5, 0.5) {$\str{A} :=$};
        \draw (2, 0) node[ptrond, label=center:\small{$\domelema$}] (a) {};
        \draw (4, 0) node[ptrond, label=center:\small{$\domelemb$}] (b) {};
        \draw (6, 0) node[ptrond, label=center:\small{$\domelemc$}] (c) {};

        \path[->] (a) edge node[yshift=-6] {\( \relsymbolR \)} (b); 
        \path[->] (b) edge node[yshift=-6] {\( \relsymbolR \)} (c); 
        \path[->] (a) edge[bend left=30] node[yshift=6] {\( \relsymbolR \)} (c); 
        \path[->] (c) edge[loop above] node[] {\( \relsymbolR \)} (c); 

        \node[] at (7.5, 0.5) {$\str{B} :=$};
        \draw (8, 0) node[ptrond, label=center:\small{$1$}] (one) {};
        \draw (10, 0) node[ptrond, label=center:\small{$2$}] (two) {};
        \draw (12, 0) node[ptrond, label=center:\small{$3$}] (three) {};

        \path[->] (one) edge node[yshift=-6] {\( \relsymbolR \)} (two); 
        \path[->] (two) edge node[yshift=-6] {\( \relsymbolR \)} (three); 
        \path[->] (three) edge[loop above] node[] {\( \relsymbolR \)} (three); 

        \draw (2, -0.5) node[label=center:\small{$\relsymbolP_1$}] (below_a) {};

        \draw (4, -0.5) node[label=center:\small{$\relsymbolP_1$}] (below_b) {};

        \draw (6, -0.5) node[label=center:\small{$\relsymbolP_1, \relsymbolP_2$}] (below_c) {};

        \draw (8, -0.5) node[label=center:\small{$\relsymbolQ_1$}] (below_one) {};

        \draw (12, -0.5) node[label=center:\small{$\relsymbolQ_2$}] (below_three) {};
    \end{tikzpicture}
  \end{figure}

  Let us first verify that $\str{A}$ can be expanded to a $\sig(\varphi)$-model $\str{A}' \models  \varphi$, and $\str{B}$ can be expanded to a $\sig(\psi)$-model $\str{B}' \models \psi$. 
  First, we can obtain a suitable extension $\str{A}'$ of $\str{A}$ by defining $\relsymbolP_1^{\str{A}'} = \{\domelema,\domelemb,\domelemc\}$ and $\relsymbolP_2^{\str{A}'} = \{ \domelemc \}$.
  To obtain a suitable extension $\str{B}'$ of $\str{B}$ we can define $\relsymbolQ_1^{\str{B}'} = \{1\}$ and $\relsymbolQ_2^{\str{B}'} = \{3\}$. 
  A straightforward calculation reveals that $\str{A}'$ and $\str{B}'$ are models of $\varphi$ and $\psi$ respectively.
  To the end of this proof we will show that $\str{A} \bisimto_{\Linfix}^{\{\relsymbolR\}} \str{B}$, which implies the non-existence of an $\Linfix[\sig(\varphi) \cap \sig(\varphi)]$ interpolant for $\varphi \models \neg \psi$, since any such interpolant $\chi$ would be true in $\str{A}$ and false in $\str{B}$.  

  To define bisimulation $\bisimZ = \bigcup_{n < \omega} \bisimZ_n$ we proceed as follows. 
  Obviously $\bisimZ_0 := \{ (\epsilon, \epsilon) \}$. 
  To define $\bisimZ_1$ we think about the game: we reply against $\domelema/\domelemb$ with $2$, against $\domelemc$ (resp. $3$) with $3$ (resp. $\domelemc$), and against the other elements with $\domelemb$.
  Thus $\bisimZ_1 = \{ (\domelema, 2), (\domelemb, 2), (\domelemc, 3), (\domelemb, 1) \}$.
  With such $\bisimZ_1$ it is immediate that \ref{bisim:back}, \ref{bisim:forth} and \ref{bisim:atomic-harmony} are satisfied for $(\epsilon, \epsilon)$.
  Then the general strategy for other rounds for the duplicator is as follows:
  \begin{itemize}\itemsep0em
  \item If the previous position was $\domelema, 1$ then reply according to the following map: $\domelema \mapsto 1$, $\domelemb / \domelemc \mapsto 2$, $1 / 3 \mapsto \domelema$, $2 \mapsto \domelemb$.
  \item If the previous position was $\domelemb, 2/3$ then reply according to the following map: $\domelema \mapsto 1$, $\domelemb / \domelemc \mapsto 3$, $1 / 2 \mapsto \domelema$, $3 \mapsto \domelemc$.
  \item If the previous position was $\domelemb, 1$ then reply according to the following map: $\domelema / \domelemb \mapsto 1$, $\domelemc \mapsto 2$, $1 / 3 \mapsto \domelemb$, $2 \mapsto \domelemc$.
  \item If the previous position was $\domelemb, 2$ or $\domelemc, 2/3$ then reply according to the following map: $\domelema / \domelemb \mapsto 2$, $\domelemc \mapsto 3$, $1 / 2 \mapsto \domelemb$, $3 \mapsto \domelemc$.
  \item If the previous position was $\domelemb, 3$ then reply according to the following map: $\domelema / \domelemb \mapsto 1$, $\domelemc \mapsto 3$, $1 / 2 \mapsto \domelemb$, $3 \mapsto \domelemc$.
  \item The position $\domelemc, 1$ is not reachable according to the above strategy.
  \end{itemize}
  It can be easily verified that playing according to the above scenarios the (non) presence of $\relsymbolR$ between the previously selected element and the current element between two structures is preserved.
  Hence, we can define a set $\bisimZ_{n{+}1}$ from $\bisimZ_n$ as $\{ (\domelemtuplea \varx \vary, \domelemtupleb \varv \varu) \}$ with $(\domelemtuplea \varx, \domelemtupleb \varv) \in \bisimZ_n$ and with $(\varx, \varu)$ ranging the aforementioned ``previous positions'' and with replies $(\vary, \varv)$ described with $\vary \mapsto \varv$ above.
  Checking if each $\bisimZ_n$ satisfies  \ref{bisim:back}, \ref{bisim:forth} and \ref{bisim:atomic-harmony} is routine. 
  The first two properties follows immediately by definition of $\vary \mapsto \varv$ above and the last property can be easily verified by hand (notice that since the vocabulary is binary, only the last two positions of tuples in $\bisimZ_n$ must be verified). We omit the boring calculations (but we did them!).

\subsection{Two-variable formulae from \texorpdfstring{$\Lsuffix$}{Lsuffix} without any \texorpdfstring{$\Linfix$}{Linfix} interpolant}\label{appendix:thm:two-variable-fragments-do-not-have-CIP}
  \begin{proof}
    Take $\varphi, \psi$ as in the proof of~\cref{thm:FL-and-FF-doesnt-have-CIP}.
    It is easy to see, by turning the innermost subformulae into DNF and by shifting quantifiers, that the following primed formulae are equivalent to their non-primed versions.
    \[
     \varphi' := [\forall{\varx_1}\; \relsymbolP_1(\varx_1) \vee \forall{\varx_2}\; [\neg \relsymbolR(\varx_1,\varx_2) \lor \forall{\varx_3}\; (\relsymbolP_2(\varx_3) \lor \neg \relsymbolR(\varx_2,\varx_3))]] \land \; \forall{\varx_1}\; [ \relsymbolP_1(\varx_1) \to \forall{\varx_2} \; \relsymbolP_2(\varx_2) \to \relsymbolR(\varx_1,\varx_2)],
    \]
    \[
    \psi' := \exists{\varx_1}( \relsymbolQ_1(\varx_1) \land \exists{\varx_2} (\relsymbolR(\varx_1,\varx_2) \land \exists{\varx_3} [ \relsymbolR(\varx_2,\varx_3) \land \relsymbolQ_2(\varx_3)])) \land \;  \forall{\varx_1} \relsymbolQ_1(\varx_1) \to \forall{\varx_2}[ \relsymbolQ_2(\varx_2) \to \neg\relsymbolR(\varx_1,\varx_2)].
    \]

    We next introduce fresh unary symbols $\relsymbolA, \relsymbolB$ and use them to get rid of subformulae starting with $\mathcal{Q}{\varx_3}$ as follows:
    \[
      \varphi_{\relsymbolA} := \forall{\varx_1}\; [ \relsymbolA(\varx_1) \leftrightarrow (\forall{\varx_2}\; (\relsymbolP_2(\varx_2) \lor \neg \relsymbolR(\varx_1,\varx_2))], \qquad
      \psi_{\relsymbolB} := \forall{\varx_1}\; [ \relsymbolB(\varx_1) \leftrightarrow \exists{\varx_2} [ \relsymbolR(\varx_1,\varx_2) \land \relsymbolQ_2(\varx_2)] ]
    \]
    and hence, $\varphi'$ and $\psi'$, relying on $\varphi_{\relsymbolA}, \psi_{\relsymbolB}$, can be rewritten into:
    \[
     \varphi'' := \varphi_{\relsymbolA} \land [\forall{\varx_1}\; \relsymbolP_1(\varx_1) \vee \forall{\varx_2}\; [\neg \relsymbolR(\varx_1,\varx_2) \lor \relsymbolA(\varx_2) ]] \land \; \forall{\varx_1}\; [ \relsymbolP_1(\varx_1) \to \forall{\varx_2} \; \relsymbolP_2(\varx_2) \to \relsymbolR(\varx_1,\varx_2)],
    \]
    \[
    \psi'' := \psi_{\relsymbolB} \land \exists{\varx_1}( \relsymbolQ_1(\varx_1) \land \exists{\varx_2} (\relsymbolR(\varx_1,\varx_2) \land \relsymbolB(\varx_2) )) \land \;  \forall{\varx_1} \relsymbolQ_1(\varx_1) \to \forall{\varx_2}[ \relsymbolQ_2(\varx_2) \to \neg\relsymbolR(\varx_1,\varx_2)].
    \]
    Note that both $\varphi''$ and $\psi''$ are in $\Lsuffix^2$. 
    As in the proof of~\cref{thm:FL-and-FF-doesnt-have-CIP} we have that $\varphi'' \models \neg \psi''$. Indeed, if $\str{A} \models \varphi''$, then either we have that $\str{A} \not\models \psi_B$, in which case also $\str{A} \not\models \psi''$, or $\str{A} \models \psi_B$, in which case also $\str{A} \not\models \psi''$, since $\varphi''$ entails that $R$ must be transitive while $\psi''$ states that this is not the case.

    We now show that there is no $\Linfix[\{ \relsymbolR \}]$-interpolant $\chi$ satisfying $\varphi'' \models \chi \models \neg \psi''$ by presenting two $\Linfix[\{ \relsymbolR \}]$-bisimilar structures satisfying $\varphi''$ and~$\psi''$. 
    Take $\str{A}''$ and $\str{B}''$ as depicted below.

    
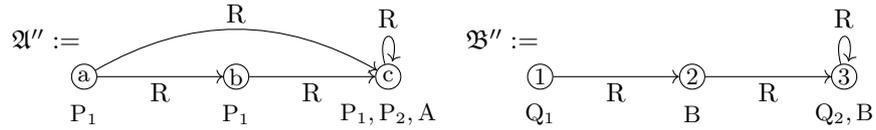
\begin{figure}[H]
      \centering
      \begin{tikzpicture}[transform shape]
          \node[] at (1.5, 0.5) {$\str{A}'' :=$};
          \draw (2, 0) node[ptrond, label=center:\small{$\domelema$}] (a) {};
          \draw (2, -0.5) node[label=center:\small{$\relsymbolP_1$}] (below_a) {};
          \draw (4, 0) node[ptrond, label=center:\small{$\domelemb$}] (b) {};
          \draw (4, -0.5) node[label=center:\small{$\relsymbolP_1$}] (below_b) {};
          \draw (6, 0) node[ptrond, label=center:\small{$\domelemc$}] (c) {};
          \draw (6, -0.5) node[label=center:\small{$\relsymbolP_1, \relsymbolP_2, \relsymbolA$}] (below_c) {};

          \path[->] (a) edge node[yshift=-6] {\( \relsymbolR \)} (b); 
          \path[->] (b) edge node[yshift=-6] {\( \relsymbolR \)} (c); 
          \path[->] (a) edge[bend left=30] node[yshift=6] {\( \relsymbolR \)} (c); 
          \path[->] (c) edge[loop above] node[] {\( \relsymbolR \)} (c); 

          \node[] at (7.5, 0.5) {$\str{B}'' :=$};
          \draw (8, 0) node[ptrond, label=center:\small{$1$}] (one) {};
          \draw (8, -0.5) node[label=center:\small{$\relsymbolQ_1$}] (below_one) {};
          \draw (10, 0) node[ptrond, label=center:\small{$2$}] (two) {};
          \draw (10, -0.5) node[label=center:\small{$\relsymbolB$}] (below_two) {};
          \draw (12, 0) node[ptrond, label=center:\small{$3$}] (three) {};
          \draw (12, -0.5) node[label=center:\small{$\relsymbolQ_2, \relsymbolB$}] (below_three) {};

          \path[->] (one) edge node[yshift=-6] {\( \relsymbolR \)} (two); 
          \path[->] (two) edge node[yshift=-6] {\( \relsymbolR \)} (three); 
          \path[->] (three) edge[loop above] node[] {\( \relsymbolR \)} (three); 
      \end{tikzpicture}
    \end{figure}
    It can be verified easily that $\str{A}'' \models \varphi''$ and $\str{B}'' \models \psi''$. 
    Moreover, the $\{ \relsymbolR \}$-reducts of $\str{A}''$ and $\str{B}''$ are exactly $\str{A}$ and $\str{B}$ from~\cref{appendix:thm:FL-and-FF-doesnt-have-CIP}, which were shown to be $\Linfix[\{ \relsymbolR \}]$-bisimilar. 
    This concludes the proof.
  \end{proof}

\subsection{HAHs in ``Base case: Step I'' of~\cref{thm:Gaffix-have-CIP}}\label{appendix:HAHs-are-important-xD}

\begin{claim}
If $\domelemd_k$ occurs in $\domelemtuplec$, then $(\domelemd_1,\ldots,\domelemd_k)$ is an affix of $\domelemtuplec$.
\end{claim}
\begin{proof}
Suppose that $d_k = c_\ell$. 
If $k = 1$, then there is nothing to show, so we assume that $k > 1$. 
Now we must also have that $\ell > 1$, since otherwise we would have that $\domelemtupled \reijosorder \domelemtuplec$ as $d_1 \lesslex c_1$ (since the length of $d_{k{-}1}$ --- and hence also the length of $d_1$ --- is strictly less than the length of $\domelemd_k = \domelemc_\ell$). 
Due to the definition of HAH, there exist $n_k, n_\ell \in \N$ such that $\domelemd_k = \domelemd_{k{-}1} \cdot n_k$ and $\domelemc_\ell = \domelemc_{\ell {-} 1} \cdot n_\ell$. 
Since $\domelemd_k = \domelemc_k$, we must have that $\domelemd_{k{-}1} \cdot n_k = \domelemc_{\ell-1} \cdot n_\ell$, which clearly implies that $\domelemd_{k-1} = \domelemc_{\ell-1}$. 
By repeating the above argument sufficiently many times we arrive at the desired conclusion.
\end{proof}

\end{document}